\documentclass[11pt]{article}
\usepackage[margin=1in]{geometry}
\usepackage{appendix}
\usepackage{babel}
\usepackage{amsthm}
\usepackage{graphicx}
\usepackage{amsmath}
\usepackage{amssymb}
\usepackage{array}
\usepackage[colorlinks=true, citecolor=black, linkcolor=black]{hyperref}
\usepackage{authblk}

\usepackage[linesnumbered,ruled,vlined]{algorithm2e}
\usepackage{algorithmicx}
\usepackage{bm}
\usepackage{tikz}
\usepackage{soul}
\usepackage{subcaption}
\usepackage{dsfont}
\usepackage{tabulary}
\usepackage{longtable}
\usepackage{url}
\usepackage{cleveref}

\newtheorem{theorem}{Theorem}[section]
\newtheorem{definition}[theorem]{Definition}

\newtheorem{lemma}[theorem]{Lemma}
\newtheorem{proposition}[theorem]{Proposition}
\newtheorem{corollary}[theorem]{Corollary}
\newtheorem{observation}[theorem]{Observation}
\newtheorem{remark}[theorem]{Remark}

\newtheorem{fact}[theorem]{Fact}
\newtheorem{example}[theorem]{Example}

\newcommand{\noindentparagraph}[1]{\noindent\textbf{#1.}}
\renewcommand{\qedhere}{}

\newcommand\class[1]{$\mathsf{#1}$}

\newcommand{\One}{\mathds{1}}

\newcommand{\config}[1]{\bm{#1}}

\newcommand{\sentence}{\Psi}
\newcommand{\pred}[1]{\mathsf{pred}\left( #1 \right)}
\newcommand{\fomodels}[2]{\mathcal{M}_{#1, #2}}
\newcommand{\propmodels}[1]{\mathcal{M}_{#1}}
\newcommand{\weight}{w}
\newcommand{\negweight}{\overline{w}}
\newcommand\wfomc{\mathsf{WFOMC}}
\newcommand\wmc{\mathsf{WMC}}

\newcommand{\symsentence}{\sentence_R}

\newcommand\FOtwo{$\text{FO}^2$}

\newcommand\Ctwo{$\text{C}^2$}

\newcommand\wcp{WCP}
\newcommand\wcpfull{Weak Connectedness Polynomial}
\newcommand\scp{SCP}
\newcommand\scpfull{Strong Connectedness Polynomial}
\newcommand\sscp{SSCP}
\newcommand\sscpfull{Strict Strong Connectedness Polynomial}
\newcommand\nscp{NSCP}
\newcommand\nscpfull{Non-strict Strong Connectedness Polynomial}

\newcommand*{\PICSIZE}{8}

\newcommand{\todo}[1]{\textcolor{blue}{[TODO: #1]}} 

\title{Bridging Weighted First Order Model Counting and Graph Polynomials\thanks{Authors appear in strict alphabetical order.}}
\author[1]{Qipeng Kuang}\affil{The University of Hong Kong, Hong Kong, China}
\author[2]{Ond\v{r}ej Ku\v{z}elka}\affil{Czech Technical University in Prague, Prague, Czech Republic}
\author[3]{Yuanhong Wang\thanks{Part of the work was done while the author was at Beihang University, China.}}\affil{Jilin University, Changchun, China}
\author[4]{Yuyi Wang}\affil{CRRC Zhuzhou Institute, Zhuzhou, China}
\date{}

\begin{document}

\maketitle

\begin{abstract}
The Weighted First-Order Model Counting Problem (WFOMC) asks to compute the weighted sum of models of a given first-order logic sentence over a given domain. It can be solved in time polynomial in the domain size for sentences from the two-variable fragment with counting quantifiers, known as \Ctwo{}. This polynomial-time complexity is known to be retained when extending \Ctwo{} by one of the following axioms: linear order axiom, tree axiom, forest axiom, directed acyclic graph axiom or connectedness axiom. An interesting question remains as to which other axioms can be added to the first-order sentences in this way. We provide a new perspective on this problem by associating WFOMC with graph polynomials. Using WFOMC, we define \wcpfull{} and \scpfull{}s for first-order logic sentences. It turns out that these polynomials have the following interesting properties. First, they can be computed in polynomial time in the domain size for sentences from \Ctwo{}. Second, we can use them to solve WFOMC with all of the existing axioms known to be tractable as well as with new ones such as bipartiteness, strong connectedness, having $k$ connected components, etc. Third, the well-known Tutte polynomial can be recovered as a special case of the Weak Connectedness Polynomial, and the Strict and Non-Strict Directed Chromatic Polynomials can be recovered from the \scpfull{}s.
\end{abstract} 

\section{Introduction}

Given a function-free first-order logic sentence $\Psi$ and $n \in \mathbb{N}$, the \textit{first-order model counting} (FOMC) problem asks to compute the number of models of $\Psi$ over the domain $\{1, \dots, n\}$. FOMC can be used, among others, to solve various combinatorics problems on labeled structures. For instance, the FOMC of the sentence $\Psi = \forall x \exists y (E(x,y) \lor E(y,x)) \land \forall x \lnot E(x,x)$ corresponds to the number of $n$-vertex directed graphs with no isolated vertices and no loops. Similarly, the FOMC of $\Phi = \forall x \lnot E(x,x) \land \forall x \forall y \left(E(x,y) \to E(y,x) \right) \land \forall x \exists_{=k} y E(x,y)$ equals the number of $k$-regular graphs on $n$ vertices where $\exists_{=k} y$ is a counting quantifier restricting the number of $y$ satisfying $E(x,y) = \top$ to be exactly $k$ for each $x$.

The weighted variant of the FOMC problem, known as the \textit{weighted first-order model counting} (WFOMC) problem, additionally expects a specification of positive and negative weights for predicates in the language, which are then used to assign weights to models. In WFOMC, the task is to compute the sum of the weights of the given sentence's models.\footnote{Surprisingly, WFOMC is also essential for the unweighted FOMC problem since, for example, Skolemization for FOMC \cite{WFOMC-FO2}, which is used to eliminate existential quantifiers, relies on the usage of (negative) weights.} WFOMC has applications in Statistical Relational Learning (SRL) \cite{getoor2007introduction}. For instance, Markov Logic Networks (MLNs) \cite{Richardson2006}, one of the most popular SRL formalisms, model relational data by a distribution over possible worlds, where the computation of the partition function (used in probabilistic inference in MLN) is reducible to WFOMC~\cite{WFOMC-UFO2}.

A key issue regarding WFOMC is its computational complexity. Generally, there is no hope of devising algorithms for WFOMC that would scale polynomially with the size of the first-order logic sentences under reasonable assumptions \cite{WFOMC-FO3}. This is true even for very simple fragments such as the one-variable fragment $\text{FO}^1$, which contains sentences with at most one logical variable. Therefore, most of the focus in this area has been on identifying fragments of first-order logic for which WFOMC can be computed in time polynomial in the size of the domain (but not necessarily in the size of the sentence). The term coined for this kind of tractability by \cite{WFOMC-UFO2} is \emph{domain liftability}---it is an analog to data complexity~\cite{datacomplexity} from database theory.
%Domain liftability appears in the context of enumerative combinatorics as well, referring to the polynomial-time countability of structures when specifying their size. \kqp{Is there really a term ``domain liftability'' in enumerative combinatorics?}

The first positive results regarding the tractability of WFOMC came from the two seminal papers \cite{WFOMC-UFO2,WFOMC-FO2}, which together established domain liftability of the two-variable fragment of first-order logic \FOtwo{}. This was quickly complemented by a hardness result \cite{WFOMC-FO3}, showing that WFOMC for the three-variable fragment $\text{FO}^3$ is not domain-liftable (under plausible complexity-theoretic assumptions). However, this does not mean that the frontiers of tractability cannot be pushed beyond \FOtwo{}. For instance, the two-variable fragment with counting quantifiers \Ctwo{} \cite{gradel1997two} (e.g., the above sentence of a $k$-regular graph) with cardinality constraints (i.e., the restrictions to the number of true ground atoms for relations), which is strictly more expressive than \FOtwo{}, was shown to be domain-liftable in \cite{WFOMC-C2}.
Other works \cite{WFOMC-conditioning,WFOMS-FO2,WFOMS-FO2-new} proved domain-liftability when attaching ground unary literals (as part of the input along with the domain size) to sentences from these fragments.

One strategy for discovering new tractable fragments is to add extra axioms to a fragment known to be tractable, which may or may not be finitely expressible in first-order logic. This approach was first explored in \cite{DBLP:conf/lics/KuusistoL18} where it was shown that adding a single functionality axiom to \FOtwo{} results in a domain-liftable fragment.\footnote{This fragment is contained in \Ctwo{}, which can encode an arbitrary number of functionality constraints.}
%, but at that time \Ctwo{} was not yet known to be domain-liftable.}
Subsequent works extended domain liftability to the fragments \Ctwo{} + $\textit{Tree}$ \cite{WFOMC-tree-axioms} and \Ctwo{} + $\textit{LinearOrder}$ \cite{WFOMC-linearorder-axiom}, which are obtained by adding to \Ctwo{} an axiom specifying that a distinguished relation $R$ should correspond to an undirected tree or to a linear order, respectively. Recently, the work \cite{WFOMC-axioms} added several new tractable classes: \Ctwo{} + $\textit{DAG}$, \Ctwo{} + $\textit{Connected}$, and \Ctwo{} + $\textit{Forest}$, which allow restricting a distinguished binary relation to be a directed acyclic graph or to be a connected undirected graph or to be a forest. Except for the linear-order and the functionality axiom, these axioms are not finitely expressible in first-order logic.

To add the axioms discussed above to \Ctwo{} while preserving domain-liftability, existing works usually rely on directly translating specialized techniques from enumerative combinatorics. For instance, Kirchhoff's matrix-tree theorem was used for the tree axiom in \cite{WFOMC-tree-axioms}. Similarly, ideas used in existing combinatorial proofs for counting connected graphs and acyclic graphs on $n$-vertices were adapted to the WFOMC setting in \cite{WFOMC-axioms} to implement the connectedness and acyclicity axioms. This raises a natural question:
\begin{quotation}
  \itshape
  Can we define a general framework for adding axioms to \Ctwo{} that allows us to prove domain-liftability of new axioms and to efficiently compute the WFOMC of sentences with these axioms?
% What axioms can be attached to the first-order logic fragments while preserving domain-liftability?
\end{quotation}

\subsection{Our Contribution}

% \kqp{I feel like we can further merge the contents, removing paragraphs.}

In this paper, we take a route inspired by \emph{graph polynomials} such as the Tutte polynomial \cite{tuttepoly} and directed chromatic polynomials \cite{chromatic-directed}. Graph polynomials are algebraic expressions that encode various combinatorial properties of graphs, allowing us to capture information such as connectivity, colorability, and spanning substructures. In this work, we consider \Ctwo{} possibly extended by cardinality constraints and ground unary literals (which we will call \Ctwo{} with cardinality constraints for simplicity).
% \footnote{\lnote{Indeed, ground unary literals are included in \Ctwo{} by default, but we mention them explicitly here as they are not fixed in the domain-liftability complexity.}}
We introduce similar polynomials for first-order logic sentences, capturing certain combinatorial properties of their models, and show that they can be computed efficiently for \Ctwo{} with cardinality constraints, which requires the introduction of new algorithmic techniques for computation of WFOMC. This, in turn, allows us to extend the domain-liftability to WFOMC for several new axioms (as well as for all of the old ones) and enables the efficient computation of graph polynomials for graphs that can be encoded by certain first-order fragments.

Using the new polynomials for WFOMC, we can tackle the problem of efficiently computing WFOMC of \Ctwo{} sentences with cardinality constraints and one of a number of non-trivial \emph{axioms}, special constraints on the graph interpreting a distinguished binary relation $R$ (denoted by $G(R)$).
%that it should be a certain combinatorial structure.
%As our polynomials allow us to obtain knowledge of combinatorial structures using similar operations as graph polynomials, this gives rise to the domain-liftability of new axioms (such as the bipartite axiom), generalization of the existing axiom (such as $connected_k$ from the connected axiom), and incorporation of the existing axioms (such as the tree axiom).
The main axioms for which we prove domain-liftability are listed in \Cref{tab:axioms}.

\begin{table}[t]
  \centering
  \caption{Axioms that we proved the domain-liftability of. Note that the first 4 axioms require the distinguished relations to be symmetric and irreflexive.}\label{tab:axioms}
    \begin{tabular}{l|p{0.4\textwidth}|l|p{0.2\textwidth}}
      \hline
      \textbf{Axiom} & \textbf{Description} & \textbf{Lifted by} & \textbf{First lifted by} \\
      \hline\hline
      $connected_k(R)$ & $G(R)$ has $k$ connected components & \Cref{exmp:k-connected-component} & new ($connected_1(R)$ by \cite{WFOMC-axioms}) \\
      \hline
      $bipartite(R)$ & $G(R)$ is a bipartite graph & \Cref{exmp:bipartite} & new \\
      \hline
      $tree(R)$ & $G(R)$ is a tree & \Cref{exmp:tree} & \cite{WFOMC-tree-axioms} \\
      \hline
      $\mathit{forest}(R)$ & $G(R)$ is a forest & \Cref{exmp:forest} & \cite{WFOMC-axioms} \\
      %\hline
      %$SS(R_1,R_2)$ & $G(R_1)$ is a spanning subgraph of $G(R_2)$; $G(R_2)$ has the same number of connected components with $G(R_1)$ & \Cref{exmp:samecc} &  new \\
      \hline\hline
      $SC(R)$ & $G(R)$ is strongly connected & \Cref{exmp:strconnected} & new \\
      \hline
      $AC(R)$ & $G(R)$ is acyclic & \Cref{exmp:acyclicity} & \cite{WFOMC-axioms} \\
      \hline
      $DT(R, Root)$ & $G(R)$ is a directed tree such that the root interprets the unary predicate $Root$ as true & \Cref{thm:dt-df} & \cite{WFOMC-tree-axioms} \\ \hline
      $DF(R)$ & $G(R)$ is a directed graph such that each weakly connected component is a directed tree & \Cref{thm:dt-df} & \cite{WFOMC-axioms} \\ \hline
      $LO(R)$ & $G(R)$ represents a linear order & \Cref{exmp:linearorder} & \cite{WFOMC-linearorder-axiom} \\
      \hline
    \end{tabular}
\end{table}

\begin{theorem}
The fragment of \Ctwo{} with cardinality constraints and any single axiom in \Cref{tab:axioms} is domain-liftable.
%\lnote{macro will give us with ... with ... here.}
\end{theorem}

Moreover, we can show domain-liftability of combinations of axioms. \Cref{tab:axiom-combination} offers a comprehensive list of axioms obtained by combining the existing axioms, all of which are new axioms that have not been shown to be domain-liftable before.\footnote{This would be difficult to achieve using existing techniques based on individual treatment of each axiom, as done in the literature.}
% A common challenge in prior research was the individual treatment of axioms using specific techniques, which made their combination a non-trivial task. However, our new approach seamlessly integrates these axioms through a unified polynomial that encapsulates the structural characteristics of the graphs.

\begin{table}[t]
    \centering
    \caption{New axioms by combining existing axioms. Note that the first 2 axioms require the distinguished relations to be symmetric and irreflexive.}
    \label{tab:axiom-combination}
    \begin{tabular}{l|p{0.45\textwidth}|l}
    % {
    %     l|
    %     >{\arraybackslash}m{0.45\textwidth}|
    %     l|l
    % }
    \hline
    \textbf{Axiom}        & \textbf{Description} & \textbf{Combined axioms} \\ \hline\hline
    $bipartite_k(R)$ &
    $G(R)$ is a bipartite graph with $k$ connected components & $connected_k \land bipartite$ \\ \hline
    $forest_k(R)$ & $G(R)$ is a forest with $k$ trees & $connected_k \land forest$ \\ \hline
    $AC_k(R)$ & $G(R)$ is a DAG with $k$ weakly connected components & $connected_k \land AC$ \\ \hline
    $BiAC(R)$ & $G(R)$ is a DAG such that the underlying undirected graph is bipartite & $bipartite \land AC$ \\ \hline
    $polytree(R)$ & $G(R)$ is a polytree, i.e., a DAG whose underlying undirected graph is a tree & $tree \land AC$ \\ \hline
    $polyforest(R)$ & $G(R)$ is a polyforest, i.e., a DAG whose underlying undirected graph is a forest & $forest \land AC$ \\ \hline
    \end{tabular}
    % \raggedright\footnotesize{$^*$ Please discriminate the polytree (resp. the polyforest) axiom from the directed tree (resp. the directed forest) axiom in \Cref{tab:axioms}.
    % : directed trees are DAGs with exactly one node with in-degree zero and with all the other nodes having in-degree $1$, and directed forests are DAGs such that every node has in-degree at most $1$. \lnote{do we need to mention this?}However, the directed tree and forest axioms can be easily lifted by our approach.
    % }
\end{table}

\begin{theorem}
The fragment of \Ctwo{} with cardinality constraints and any single axiom in \Cref{tab:axiom-combination} is domain-liftable.
\end{theorem}

%Our approach offers several distinct advantages over existing techniques for individual axioms. First, it is more general, requiring only a polynomial representation for axioms without relying on specific algorithmic techniques for each problem. Second, it enables the use of the same polynomial for multiple axioms. Third, it is more flexible, allowing easy extension to new axioms as long as they can be expressed with polynomials, for instance, requiring $G(R)$ to represent a $k$-cycle permutation in Example \ref{ex:permk} or a strongly connected tournament in Example \ref{exmp:sct}. Lastly, it establishes a connection between the complexity of graph polynomials and the domain-liftability of WFOMC.

%The second task is inspired by SRL problems that have constraint for the model in other forms. For example, in some applications the weight of a model depends on the number of connected components of $G(R)$ for some distinguished relation $R$, which is exactly the indicator to classify the models by computing its \wcp{} and therefore we can build up the model for the problem. Such example is elaborated in \Cref{exmp:expweight}. \kqp{Maybe introduce this example in a better way.}

% \noindentparagraph{Contribution to Graph Polynomials}
Besides using the idea of graph polynomials to solve several counting tasks efficiently, conversely, the polynomials for WFOMC bring new results on graph polynomials. We show that the Tutte polynomial \cite{tuttepoly}, the strict and non-strict directed chromatic polynomials \cite{chromatic-directed} can be recovered from our polynomials for WFOMC. This allows us to show that these graph polynomials can be computed in time polynomial in the number of vertices for any graph that can be encoded by a set of ground unary literals, cardinality constraints and a fixed \Ctwo{} sentence (the parameters of cardinality constraints and the number of ground unary literals are allowed to depend on the number of vertices). This result answers in the positive way the question in \cite{WFOMC-tree-axioms} whether calculating the Tutte polynomial of a block-structured graph (see its definition in Example~\ref{ex:block-structured-graph}) can be done in time polynomial in the number of vertices.

\subsection{Related Work}

We build on several lines of research. First, we build on the stream of results from lifted inference literature which established domain-liftability of several natural fragments of first-order logic \cite{poole2003first,DBLP:conf/ijcai/BrazAR05,WFOMC-UFO2,DBLP:conf/uai/GogateD11a,WFOMC-FO2,WFOMC-FO3,DBLP:conf/nips/KazemiKBP16,WFOMC-C2}. These works started with the seminal results showing domain-liftability of \FOtwo{} \cite{WFOMC-UFO2,WFOMC-FO2} and continued by showing domain-liftability of larger fragments such as S$^2$FO$^2$ and S$^2$RU \cite{DBLP:conf/nips/KazemiKBP16}, U$_1$ \cite{DBLP:conf/lics/KuusistoL18} or \Ctwo{} \cite{WFOMC-C2}. Second, we build on extensions of these works that added axioms not expressible in these fragments or even in first-order logic such as the tree axiom \cite{WFOMC-tree-axioms}, forest axiom \cite{WFOMC-axioms}, linear order axiom \cite{WFOMC-linearorder-axiom}, acyclicity axiom \cite{WFOMC-axioms} and connectedness axiom \cite{WFOMC-axioms}. In this work, we generalize all of these results and add new axioms to this list, showing their domain-liftability. Finally, our work builds on tools from enumerative combinatorics \cite{stanley1986enumerative} and on graph polynomials~\cite{tutte2004graph}. The most well-known among graph polynomials is the Tutte polynomial \cite{tuttepoly} but there exist also extensions of it \cite{tutte-directed}, with which we do not work directly in this paper. Computing the Tutte polynomial is known to be intractable in general \cite{tuttepoly-hardness}, however, in this work we show that it can be computed in polynomial time for graphs that can be encoded by a fixed sentence from the first-order fragment \Ctwo{}, possibly with cardinality constraints and ground unary literals, both of which may depend on the number of vertices of the graph (unlike the \Ctwo{} sentence which needs to stay fixed).
% \lnote{may need an update here.}

%\subsection{Organization} 

\section{Preliminaries}

We introduce our notations and main technical concepts, and provide a brief overview of existing algorithms in this section.

\subsection{Mathematical Notations}

Throughout this paper, we denote the set of natural numbers from 1 to $n$ as $[n]$. For a polynomial $f(x)$, the notation $[x^i] f$ represents the coefficient of the monomial $x^i$.

All vectors used in this paper are non-negative integer vectors. For a vector $\config{c}$ of length $p$, $|\config{c}|$ denotes the sum of its elements, i.e., $\sum_{i=1}^p c_i$. When considering two vectors $\config{a}$ and $\config{b}$ of the same length $p$, the operation $\config{a} + \config{b}$ denotes the vector $(a_1+b_1, a_2+b_2, \cdots, a_p+b_p)$, and the subtraction is defined similarly. For a vector $\config{a}$ and a non-negative integer $n$, $\binom{n}{\config{a}}$ refers to $\binom{n}{a_1} \binom{n-a_1}{a_2} \cdots \binom{n-a_1-\cdots-a_{p-1}}{a_p} = \frac{n!}{a_1!a_2! \cdots a_p!}$.

\subsection{First-Order Logic}\label{sec:fol}

We consider the function-free (we allow constants but not any functions with arity greater than zero), finite domain fragment of first-order logic. An \emph{atom} of arity $k$ takes the form $P(x_1, \cdots, x_k)$ where $P/k$ is from a vocabulary of \emph{predicates} (also called \emph{relations}), and $x_1, \cdots, x_k$ are logical variables from a vocabulary of variables or constants from a vocabulary of constants.\footnote{We restrict $k \ge 1$ for simplicity but $k \ge 0$ is tractable for this work as well.} A \emph{literal} is an atom or its negation. A \emph{formula} is a literal or formed by connecting one or more formulas together using negation, conjunction, or disjunction. A formula may be optionally surrounded by one or more quantifiers of the form $\forall x$ or $\exists x$, where $x$ is a logical variable. A logical variable in a formula is said to be \emph{free} if it is not bound by any quantifier. A formula with no free variables is called a \emph{sentence}.
% The equality symbol ``='' is not allowed.\footnote{Though, the identity relation can be expressed by cardinality constraints as shown at the end of \Cref{sub:data-complexity}.}

A first-order logic formula in which all variables are substituted by constants in the domain is called \emph{ground}. A \emph{possible world} $\omega$ interprets each predicate in a sentence over a finite domain, represented by a set of ground literals. Sometimes we only care about the interpretation of a single predicate $R$ in a possible world $\omega$, which is denoted by $\omega_R$. The satisfaction relation $\models$ is defined in the usual way: $\omega \models \alpha$ means that the formula $\alpha$ is true in $\omega$. The possible world $\omega$ is a \emph{model} of a sentence $\sentence$ if $\omega \models \sentence$. We denote the set of all models of a sentence $\sentence$ over the domain $\{1,2,\dots,n\}$ by $\fomodels{\sentence}{n}$.

\subsection{Weighted First Order Model Counting}\label{sec:wfomc}

The \textit{weighted first-order model counting problem} (WFOMC) takes the input consisting of a first-order sentence $\sentence$, a domain size $n\in\mathbb{N}$, and a pair of weighting functions $(\weight, \negweight)$ that both map all predicates in $\sentence$ to real weights.
% ~\footnote{We consider the \emph{symmetric} WFOMC, in which the weight of a model is independent of the domain elements in the model.}.
Given a set $L$ of ground literals whose predicates are in $\sentence$, the weight of $L$ is defined as $W(L, \weight, \negweight):= \prod_{l \in L_T}\weight(\pred{l}) \cdot \prod_{l \in L_F}\negweight(\pred{l}),$ where $L_T$ (resp. $L_F$) denotes the set of positive (resp. negative) ground literals in $L$, and $\pred{l}$ maps a literal $l$ to its corresponding predicate name.
\begin{example}\label{exmp:model-weight}
    Consider the sentence $\sentence = \forall x \forall y \left( R(x,y) \to S(y) \right)$ and the weighting functions $\weight(R) = 2, \negweight(R) = \weight(S) = \negweight(S) = 1$.
    The weight of the atoms set
    $$L = \{R(1,1), \lnot R(1,2), R(2,1), R(2,2), S(1), S(2)\}$$
    is $w(R) \cdot \negweight(R) \cdot \weight(R) \cdot \weight(R) \cdot \weight(S) \cdot \weight(S) = 2 \cdot 1 \cdot 1 \cdot 1 \cdot 1 \cdot 1 = 2$.
\end{example}
\begin{definition}[Weighted First Order Model Counting]
    The WFOMC of a first-order sentence $\sentence$ over a finite domain of size $n$ under weighting functions $(\weight, \negweight)$ is
    \begin{equation*}
        \wfomc(\sentence, n, \weight, \negweight) := \sum_{\mu \in \fomodels{\sentence}{n}} W(\mu, \weight, \negweight).
    \end{equation*}
\end{definition}

Since these weightings are defined on the predicate level, all positive ground literals of the same predicate get the same weights, and so do all negative ground literals of the same predicate.
For this reason, the notion of WFOMC defined here is also referred to as \emph{symmetric} WFOMC \cite{WFOMC-FO3}.

\begin{example}
    Consider $\sentence = \forall x\exists y \ R(x,y)$, and let $\weight(R) = 1, \negweight(R) = 1$.
    Then $\wfomc(\sentence, n, \weight, \negweight) = (2^n - 1)^n$ because for each domain element $i\in[n]$, there are $2^n - 1$ ways to choose the truth assignment of $R(i,1),R(i,2),\dots,R(i,n)$ that make $\exists y \  R(i, y)$ satisfied, and there are $n$ domain elements.
    For another example, consider the sentence $\sentence$ and the weighting functions $(\weight,\negweight)$ in Example \ref{exmp:model-weight}.
    One may check that $\wfomc(\sentence, n, \weight, \negweight) = (3^n + 1)^n$.
\end{example}

The \emph{weighted model counting problem} (WMC) is defined similarly but the input formula is a ground quantifier-free formula.
\begin{definition}[Weighted Model Counting]
    The WMC of a ground quantifier-free formula $\Phi$ under weighting functions $(\weight, \negweight)$ is
    \begin{equation*}
        \wmc(\Phi, \weight, \negweight) := \sum_{\mu \in \propmodels{\Phi}} W(\mu, \weight, \negweight),
    \end{equation*}
    where $\propmodels{\Phi}$ denotes the set of all models of $\Phi$.
\end{definition}
\begin{example}
    The WMC of a ground quantifier-free formula $\Phi = R(a,b) \to S(a)$ under $\weight(R) = 2, \negweight(R) = \weight(S) = \negweight(S) = 1$ is
    $$W(\{R(a,b), S(a)\}, \weight, \negweight) + W(\{\lnot R(a,b), S(a)\}, \weight, \negweight) + W(\{\lnot R(a,b), \lnot S(a)\}, \weight, \negweight) = 4.$$
\end{example}

% For example, the WFOMC of $\sentence = \forall x \forall y \left( R(x,y) \to S(y) \right)$ over a domain of size $n$ under $\weight(R) = 2, \negweight(R) = \weight(S) = \negweight(S) = 1$ is $\left( 3^n + 1 \right)^n$.

\subsection{Data Complexity of WFOMC}
\label{sub:data-complexity}

We consider the \textit{data complexity} of WFOMC: the complexity of computing $\wfomc(\sentence, n, \weight, \negweight)$ when fixing the input sentence $\sentence$ and weighting $(\weight, \negweight)$, and treating the domain size $n$ as an input encoded in unary.

\begin{definition}[Domain-liftability \cite{WFOMC-UFO2}]
A sentence, or class of sentences, that enjoys polynomial-time data complexity, i.e., the complexity is polynomial in the domain size, is said to be \emph{domain-liftable}.
\end{definition}

\noindentparagraph{Domain-liftability of \FOtwo{}}
Building on \cite{WFOMC-UFO2,WFOMC-FO2}, in ~\cite[Appendix C]{WFOMC-FO3}, it was shown that any sentence with at most two logical variables (\FOtwo{}) is domain-liftable by devising a polynomial-time algorithm for computing $\wfomc(\sentence, n, \weight, \negweight)$ for any \FOtwo{} sentence $\sentence$ and any weighting $(\weight, \negweight)$.

Here we briefly describe this algorithm, as some of its concepts are important in this paper.
We mostly follow the notation from \cite{WFOMC-FO2-faster}.
Any sentence in \FOtwo{} can be reduced to the prenex normal form with only universal quantifiers (i.e., $\forall x \forall y \ \psi(x,y)$, where $\psi(x,y)$ is a quantifier-free formula) by the normalization in \cite{scott-normal-form} and the technique of eliminating existential quantifiers in \cite{WFOMC-FO2}.
The reduction respects the WFOMC value, and hence in what follows, we assume that the input sentence is in the form $\forall x \forall y \ \psi(x,y)$.
We first need the notion of a \emph{1-type}.

\begin{definition}[1-type]
A 1-type of a first-order sentence $\sentence$ is a maximally consistent set\footnote{A set of literals is maximally consistent if it is consistent (does not contain both a literal and its negation) and cannot be extended to a larger consistent set.} of literals formed from atoms in $\sentence$ using only a single variable $x$.
\end{definition}

For example, $\sentence = \forall x \forall y \left( F(x) \land G(x,y) \right)$ has four 1-types: $F(x) \land G(x,x)$, $F(x) \land \lnot G(x,x)$, $\lnot F(x) \land G(x,x)$ and $\lnot F(x) \land \lnot G(x,x)$. Intuitively, a 1-type interprets unary and reflexive binary predicates for a single domain element.

% Suppose our task is to calculate $\wfomc(\sentence, n, \weight, \negweight)$ where $\sentence = \forall x \forall y \ \psi(x,y)$.
Let $C_1(x), C_2(x), \cdots, C_p(x)$ be all the 1-types of $\sentence$. When we consider domain-liftability, the sentence $\sentence$ is fixed and therefore $p$ is a constant. We partition domain elements into 1-types by enumerating the \emph{configuration} $\config{c} = (c_1, c_2, \cdots, c_L)$, a vector indicating the number of elements partitioned to each 1-type, since in the symmetric WFOMC only the configuration matters.
Then we can rewrite $\sentence$ in the form of
$$
\begin{aligned}
\sentence = \bigwedge_{1 \le i < j \le p} \forall x \in C_i \forall y \in C_j (\psi(x,y) \land \psi(y,x)) \land \bigwedge_{1 \le i \le p} \forall x \in C_i \forall y \in C_i \ \psi(x,y),
\end{aligned}
$$
where $\psi(x,y)$ is a quantifier-free formula.

Since we know the truth values of the unary and reflexive binary atoms of each 1-type, we may simplify the body of each conjunct by substituting every unary and reflexive binary literal with true or false as appropriate. Denote by $\psi_i(x, y)$ the simplification of $\psi(x, y)$ when both $x$ and $y$ belong to the same 1-type $C_i$, and by $\psi_{ij}(x, y)$ the simplification of $\psi(x, y) \land \psi(y, x)$ when $x$ and $y$ belong to $C_i$ and $C_j$ respectively. We then have:
$$
\begin{aligned}
\sentence = \bigwedge_{1 \le i < j \le p} \forall x \in C_i \forall y \in C_j \ \psi_{ij}(x,y) \land \bigwedge_{1 \le i \le p} \forall x \in C_i \forall y \in C_i \ \psi_{i}(x,y).
\end{aligned}
$$
An important property here is that each conjunct is independent of others since no two conjuncts involve the same ground binary literal. Let
$$
\begin{aligned}
r_{i,j} &= \wmc(\psi_{i,j}(a,b), \weight, \negweight), \\
s_{i} &= \wmc(\psi_{i}(a,b) \land \psi_{i}(b,a), \weight, \negweight)
\end{aligned}
$$
be the \emph{mutual} and \emph{internal} 1-type coefficient of $\sentence$ respectively. Then we have
\begin{equation}
    \wfomc(\sentence, n, \weight, \negweight) = \sum_{|\config{c}|=n} \binom{n}{\config{c}} \prod_{1 \le i \le p} W(C_i, \weight, \negweight)^{c_i} \cdot s_i^{\binom{c_i}{2}} \prod_{1 \le i < j \le p} r_{i,j}^{c_ic_j},
    \label{eq:wfomc-fo2}
\end{equation}
which can be calculated in polynomial time in $n$.

\noindentparagraph{Domain-liftability of \Ctwo{}}
The above algorithm also serves as the basis of a WFOMC algorithm for \Ctwo{} (the \FOtwo{} sentence with \emph{counting quantifiers}) together with \emph{cardinality constraints} and \emph{ground unary literals}.

\begin{definition}[Counting quantifiers]
The counting quantifier $\exists_{=k}$ restricts that the number of assignments of the quantified variable satisfying the subsequent formula is exactly $k$. Counting quantifiers $\exists_{\le k}$ and $\exists_{\ge k}$ are defined similarly. %\kqp{Why don't we need this sentence? ``In this paper we let $k$ be a constant.''} - because it sounded like this is a constant chosen uniformly over the whole paper. When we say that the sentence is fixed, that also means that the k=parameters of the quantifiers are fixed.
\end{definition}

Counting quantifiers allow us to express concepts such as ``each vertex in the undirected graph has exactly $k$ edges'', i.e., $k$-regular graphs, by the sentence $\forall x \lnot E(x,x) \land \forall x \forall y \left(E(x,y) \to E(y,x) \right) \land \forall x \exists_{=k} y E(x,y)$.

\begin{definition}[Cardinality constraints]\label{def:cc}
The cardinality constraint is an expression of the form of $|P| \bowtie k$, where $P$ is a predicate and $\bowtie$ is a comparison operator $\{<, \le, =, \ge, >\}$.
These constraints are imposed on the number of distinct positive ground literals of $P$ in a model.
\end{definition}

For example, $|Eq| = n$ means that there are exactly $n$ positive ground literals of $Eq$ in a model, and the sentence $(\forall x \ Eq(x,x)) \land |Eq| = n$ along with a domain of size $n$ encodes an identity relation $Eq$. Note that we allow $k$ in \Cref{def:cc} to depend on the domain size (such as in the definition of $Eq$ here where we used $|Eq| = n$).

Finally, we show the usage of ground unary literals defined in Section \ref{sec:fol}. For example, the sentence $\forall x \exists y (R(x,y) \to S(y)) \land S(1) \land \lnot S(2)$ contains ground unary literals $S(1) \land \lnot S(2)$ indicating that any interpretation of this sentence should satisfy that $S(1)$ is true and $S(2)$ is false. Another example is encoding a block-structured graph.

\begin{example}
    \label{ex:block-structured-graph}
    A \emph{block-structured graph} is an undirected graph $G=(V, E)$ where vertices are partitioned into a constant number of blocks, and for every two blocks $V_i, V_j$, either $\forall x \in V_i \forall y \in V_j, (x,y) \in E$ or $\forall x \in V_i \forall y \in V_j, (x,y) \not\in E$.
    We can encode a block-structured graph in the following way: we introduce a fresh unary predicate $Block_i$ for each block $V_i$, and let
    \begin{equation}\label{eq:block-structured-graph}
    \begin{aligned}
    \sentence_B = & \bigwedge_{V_i, V_j\text{ are connected}} \forall x \forall y (Block_i(x) \land Block_j(y) \to E(x,y)) \\
    & \land \bigwedge_{V_i, V_j\text{ are not connected}} \forall x \forall y (Block_i(x) \land Block_j(y) \to \lnot E(x,y)) \\
    & \land \bigwedge_{e \in V} \left( Block_{v(e)}(e) \land \bigwedge_{j\neq v(e)} \lnot Block_{j}(e) \right),
    \end{aligned}
    \end{equation}
    where $v(e)$ indicates the block that vertex $e$ belongs to, and $Block_*(e)$ and $\lnot Block_*(e)$ are ground unary literals for vertex $e$ indicating whether $e$ is in the specified block or not.
\end{example}

\begin{remark}
    We need to define the data complexity more carefully when the input sentence of WFOMC involves cardinality constraints and ground unary literals.
    In such instances, the \Ctwo{} sentence (including the parameters in counting quantifiers) in the input remains fixed, while cardinality constraints and ground unary literals are considered variable inputs encoded in binary.
    %The parameters of cardinality constraints and the number of ground unary literals may depend on the domain size.
     %\lnote{remove "even if the parameters of cardinality constraints and the number of ground unary literals depend on the domain size"? We have stressed this in the above remark}.
\end{remark}
% The domain-liftability\footnote{The distinction in the domain-liftability between counting quantifiers and cardinality constraints is noteworthy. In the case of counting quantifiers, the counting parameter is treated as an inherent component of the sentence and remains constant when evaluating the data complexity. Conversely, cardinality constraints allow the cardinality parameter to be included as part of the input instance, while still guaranteeing polynomial running time.} of \Ctwo{} with cardinality constraints was shown in \cite{WFOMC-C2}, where its WFOMC was first reduced to a WFOMC in \FOtwo{} with symbolic weightings, and then solved by the algorithm above.

The presence of counting quantifiers, cardinality constraints, and ground unary literals does not affect the domain-liftability of \FOtwo{}. Using the techniques in \cite{WFOMC-C2} and \cite[Appendix A]{WFOMS-FO2-new}, WFOMC of \Ctwo{} with cardinality constraints and ground unary literals can be reduced to WFOMC of \FOtwo{} with symbolic weights\footnote{The weighting functions $\weight, \negweight$ map predicates to symbolic variables. For more details, see \cite[Section $3.1$]{WFOMC-C2}.}, and then solved by the algorithm for \FOtwo{} in a straightforward manner.

\begin{remark}
The techniques mentioned above apply to the methods in this paper.
Therefore, in the rest of this paper we only present the algorithms for \FOtwo{} and state the results in terms of \Ctwo{} with cardinality constraints for simplicity.
\end{remark}

\subsection{Axioms}

While \Ctwo{} captures numerous graph structures, it falls short in expressing certain important graph properties, such as being a tree, a forest, and strongly connected.
In fact, these structures cannot be finitely axiomatized by first-order logic \cite{libkin2004elements}.% without listing all the ground literals explicitly.
%However, in terms of WFOMC, it is possible to express these properties more succinctly using what we refer to as \emph{axioms}.

Given a binary predicate $R$, an interpretation for $R$ can be regarded as a directed graph where the domain is the vertex set and $R$ is the edge set, which we call $G(R)$.
If a possible world $\omega$ is specified, the interpretation of $R$ in $\omega$ is a specific graph, denoted by $G(\omega_R)$.

An axiom is a special constraint on the interpretation of a binary predicate $R$ that $G(R)$ should be a certain combinatorial structure.
We often write an axiom as $axiom(R)$ where $axiom$ is an identifier of the axiom.
For example, in what follows, we will discuss the \emph{tree axiom} $tree(R)$, which requires $G(R)$ to be a tree, and the \emph{forest axiom} $forest(R)$, which requires $G(R)$ to be a forest.
For better illustration, we treat axioms as atomic formulas in a sentence, and use $\sentence \land axiom(R)$ to denote the sentence $\sentence$ with the axiom $axiom(R)$ conjoined.
For instance, $tree(E) \land \forall x \left( Leaf(x)\leftrightarrow \exists_{=1} y \ E(x,y) \right)$ is a sentence that expresses the concept of trees with a leaf predicate $Leaf/1$.

\subsection{Graph Polynomial}\label{sec:graphpolys}
Graph polynomials are a class of polynomials defined on graphs that carry combinatorial meaning.

\noindentparagraph{Tutte Polynomial}
The Tutte polynomial~\cite{tuttepoly} is one of the most well-known graph polynomials for undirected graphs.
Consider an undirected graph $G=(V, E)$, where $V$ is the set of vertices and $E$ is the set of edges.
%The \emph{connected components} of $G$ are the maximal connected subgraphs of $G$, i.e., the subgraphs that are connected and cannot be extended by adding more vertices or edges.
The Tutte polynomial for the graph $G$ is a polynomial in two variables $x$ and $y$ defined as
\begin{equation*}
T_G(x,y) = \sum_{A \subseteq E} (x-1)^{cc(A)-cc(E)} (y-1)^{cc(A)+|A|-|V|},
\end{equation*}
where $cc(A)$ is the number of connected components considering only edges in $A$.

The Tutte polynomial holds significant combinatorial meaning. Evaluating the polynomial at specific points reveals various combinatorial properties. For example, $T_G(1,1)$ counts the number of spanning forests of $G$ (specifically, the number of spanning trees if $G$ is connected), and $T_G(0,2)$ counts the strongly connected orientations of $G$ \cite{tuttepoly-20}. Substituting a concrete value for only one variable yields another graph polynomial of $G$. For example, $T_G(x,0)$ corresponds to the chromatic polynomial~\cite{read1968introduction} of $G$; $T_G(0,y)$ corresponds to the flow polynomial~\cite{diestel2008graph} of $G$.

The computation of the Tutte polynomial for an arbitrary graph is known to be computationally challenging.
Even the evaluation of a specific point of the polynomial is known to be \class{\#P}-hard, with only a limited number of points being computationally tractable~\cite{tuttepoly-hardness}.

\noindentparagraph{Directed Chromatic Polynomial}
%The chromatic polynomial of a directed graph \cite{chromatic-directed} is an analogue of the chromatic polynomial of an undirected graph. \kqp{need citation and put this sentence after the definition of directed version.}
The idea of the chromatic polynomial of a directed graph appeared in \cite{chromatic-directed}, although it was originally stated in terms of ordered sets.
Given a directed graph $D=(V,E)$ where $n = |V|$, let $\chi_D(x)$ be the number of ways of coloring each vertex in $D$ by one of the colors $\{1, 2, \cdots, x\}$ such that if there is an edge from $u$ to $v$, the color of $u$ is smaller than the color of $v$. We have the following lemma to state that $\chi_D(x)$ is a polynomial of $x$.

\begin{lemma}[\cite{chromatic-directed}]\label{lemma:dichromatic-surjective}
For a directed graph $D$ of size $n$, let $\chi^*_D(i)$ be the number of surjective colorings of vertices using $i$ colors (i.e., each color should be used at least once) such that if there is an edge from $u$ to $v$, the color of $u$ is smaller than the color of $v$. It holds that
\begin{equation*}
\chi_D(x) = \sum_{i=1}^n \binom{x}{i} \chi^*_D(i) = \sum_{i=1}^n \frac{x(x-1)\cdots (x-i+1)}{i!} \chi^*_D(i),
\end{equation*}
where $\binom{x}{i} = 0$ if $i > x$.
\end{lemma}

\begin{corollary}\label{corol:dichromatic-degree}
$\chi_D(x)$ is a polynomial of $x$ of degree $n$.
% Moreover, $[x^n]\chi_D(x) = \frac{\chi^*_D(n)}{n!}$.
% \lnote{Is there any reference?}\kqp{It follows from the above lemma.}
\end{corollary}

We call $\chi_D(x)$ the strict directed chromatic polynomial of $D$. A variant of it is the non-strict directed chromatic polynomial $\bar \chi_D(x)$ where for each positive integer $x$, $\bar \chi_D(x)$ equals the number of colorings of the vertices in $D$ by one of the $x$ colors $\{1, 2, \cdots, x\}$ such that if there is an edge from $u$ to $v$, the color of $u$ is smaller than or equal to the color of $v$. Similarly to $\chi_D(x)$, $\bar \chi_D(x)$ is also a polynomial of $x$.

\begin{corollary}\label{corol:nonstrictdichromatic-degree}
$\bar \chi_D(x)$ is a polynomial of $x$ of degree $n$.
% Moreover, $[x^n]\bar \chi_D(x) = \frac{\bar \chi^*_D(n)}{n!}$ where $\bar \chi^*_D(x)$ is defined analogically to $\chi^*_D(x)$ except that the color of $u$ is smaller than or equal to the color of $v$ if there is an edge from $u$ to $v$.
\end{corollary}

The relation between $\chi_D(x)$ and $\bar \chi_D(x)$ is given by the following lemma.

\begin{lemma}[\cite{chromatic-directed,tutte-directed}]\label{lemma:dichromatic-negative}
Let $D$ be a directed graph and $acyc(D)$ be the directed acyclic graph obtained from $D$ by condensing each cycle into a vertex. Denote the number of vertices in $acyc(D)$ by $|V(acyc(D))|$. We have
\begin{equation*}
\begin{aligned}
\chi_D(x) &= \begin{cases}
    (-1)^n \bar \chi_D(-x), & D \text{ is acyclic,} \\
    0, & \text{otherwise,}
    \end{cases} \\
\bar \chi_D(x) &= (-1)^{|V(acyc(D))|} \chi_{acyc(D)}(-x).
\end{aligned}
\end{equation*}
\end{lemma}

% If $D$ consists of several weakly connected components, the directed chromatic polynomial is the product of the same polynomial on each weakly connected component.

% \begin{lemma}\label{lemma:dichromatic-wcc}
% Let $D_1 = (V_1, E_1), \cdots, D_k = (V_k, E_k)$ be $k$ disjoint directed graphs and denote by $D = \cup_{i=1}^k D_i$ the directed graph $(\cup_{i=1}^k V_i, \cup_{i=1}^k E_i)$. Then,
% \begin{equation*}
%   \begin{aligned}
%     \chi_D(x) &= \prod_{i=1}^k \chi_{D_i}(x), \\
%     \bar \chi_D(x) &= \prod_{i=1}^k \bar \chi_{D_i}(x).
%   \end{aligned}
% \end{equation*}
% \end{lemma}

% \paragraph{B-Polynomial}

% The B-Polynomial \cite{tutte-directed} is an extension of Tutte Polynomial for directed graphs. For a directed graph $D = (V, E)$, the B-polynomial is the trivariate polynomial $B_D(q,y,z)$ such that for all positive integers $q$,
% \begin{equation*}
% B_D(q,y,z) = \sum_{f: V \to [q]} y^{|f^>|} z^{|f^<|},
% \end{equation*}
% where the sum is over all $q$-labellings of the vertices of $D$, $f^>$ is the set of edges $(u,v)$ where the label of $v$ is larger than the label of $u$, and $f^<$ is the set of edges $(u,v)$ where the label of $v$ is samller than the label of $u$. 

\section{\wcpfull{}}

In this section, we introduce our first polynomial, called \emph{\wcpfull{}} (\wcp{}).
%, which serves as an analogous notion to graph polynomials but is specifically tailored for first-order logic sentences.
The \wcp{} is a univariate polynomial defined on a first-order logic sentence and a distinguished binary relation.
Similarly to graph polynomials, the \wcp{} can also capture various combinatorial properties of the models of the sentence w.r.t. the binary relation.
% However, the \wcp{} distinguishes itself by its ability to be efficiently computed in polynomial time in the domain size for the \Ctwo{} fragment with cardinality constraints and ground unary literals as we show in this section.
As we show later in this section, \wcp{} can be computed in polynomial time in the domain size for sentences from the fragment \Ctwo{} with cardinality constraints. This property allows us to attach axioms of non-trivial combinatorial structures to the first-order logic sentence as well as to derive the Tutte polynomial for certain graphs.

\subsection{Definition}

We define \wcp{} as a univariate polynomial whose point evaluation at every positive integer $u$ is equal to a weighted first-order model count of a suitable first-order logic sentence parameterized by $u$.

\begin{definition}[\wcpfull{}]\label{def:wcp}
Let $\Psi$ be a first-order logic sentence possibly with cardinality constraints and ground unary literals, $\weight, \negweight$ be two weighting functions, and $R$ be a distinguished binary relation.
The $n$-th \wcp{} $f_n(u)$ of $\sentence$ for the relation $R$ is the univariate polynomial which satisfies:
\begin{equation*}
f_n(u;\Psi, \weight, \negweight, R) = \wfomc(\Psi_{R,u}, n, \weight, \negweight)
\end{equation*}
for all positive integers $u$. Here, the sentence $\Psi_{R,u}$ is defined as:
% The \wfomcpolys{} $\{f_1(u; R), f_2(u; R), \cdots \}$ of $(\sentence, \weight, \negweight)$ for the relation $R$ which satisfies
% \begin{equation*}
% f_n(u; R) = \wfomc(\Psi_{R,u}, n, \weight, \negweight)
% \end{equation*}
% for all positive integers $u$ is called the \emph{\wcpfull{}} of $\sentence, \weight, \negweight$ and $R$. Here, the sentence $\Psi_{R,u}$ is defined as
\begin{equation}
  \begin{aligned}
    \Psi_{R,u} = \Psi \wedge & \bigwedge_{i=1}^u \forall x \forall y \left( A_i(x) \wedge (R(x,y) \vee R(y,x)) \to A_i(y) \right) \\
    \wedge & \bigwedge_{i = 2}^u \forall x \left( A_i(x) \to A_{i-1}(x) \right),
  \end{aligned}
  \label{eq:wcp}
\end{equation}
where $A_1, \cdots, A_u$ are fresh unary predicates and $w(A_i) = \overline{w}(A_i) = 1$.
\end{definition}

Intuitively, by introducing $u$ new predicates $A_1, \cdots, A_u$, we make each \emph{weakly} connected component in $G(R)$ contribute a factor of $u+1$ to $\wfomc$.
In the rest of the paper, we use $cc(\mu_R)$ to denote the number of weakly connected components in $G(\mu_R)$.
The most important property of \wcp{} is that it captures $cc(\mu_R)$, which is stated formally in the following proposition.

\begin{proposition}\label{prop:wcp}
For the sentence $\Psi$, the weighting functions $w$ and $\overline{w}$ and the binary relation $R$, it holds that
$$
f_n(u; \Psi, w, \overline{w}, R) = \sum_{\mu \in \fomodels{\Psi}{n}} W(\mu, \weight, \negweight) \cdot (u+1)^{cc(\mu_R)}.
$$
\end{proposition}

\begin{proof}
By $\bigwedge_{i = 2}^u \forall x \left( A_i(x) \to A_{i-1}(x) \right)$, for each element $e$ in the domain there must be an $i' \in [0,u]$ such that $A_1(e) = \cdots = A_{i'}(e) = \top$ and $A_{i'+1}(e) = \cdots = A_u(e) = \bot$. Meanwhile, elements whose corresponding vertices form a weakly connected component in $G(\mu_R)$ must have the same $i'$ due to the first row of Equation \eqref{eq:wcp}.

Therefore, for any model $\mu$ of $\sentence$, each weakly connected component in $G(\mu_R)$ has $u+1$ choices of $i'$ regardless of the interpretations of other predicates.
That is, $\mu$ replicates $u+1$ times in $\sentence_{R,u}$, and the weight of each of these replicas is $W(\mu, w, \overline{w})$ due to $\weight(A_i) = \negweight(A_i) = 1$.
Therefore, the WFOMC of $\sentence_{R,u}$ is $\sum_{\mu\in\fomodels{\Psi}{\mu}} W(\mu, w, \overline{w}) \cdot (u+1)^{cc(\mu_R)}$, which completes the proof.
\end{proof}

\begin{proposition}\label{prop:wcp-degree}
The $n$-th \wcp{} $f_n(u; \Psi, w, \overline{w}, R)$ always exists, and it is a unique univariate polynomial of degree at most $n$.
\end{proposition}

\begin{proof}
The existence and the upper bound of the degree follow immediately from Proposition \ref{prop:wcp} since $cc(\mu_R) \le n$.
\end{proof}

\subsection{Fast Calculation}

Before applying the polynomials in any context, we need to be able to compute them (or at least evaluate them at some points) efficiently.
By Proposition \ref{prop:wcp-degree}, $n+1$ evaluations are sufficient to interpolate the $n$-th \wcp{} (e.g., by Lagrange Interpolation). Therefore, our task is to evaluate the \wcp{} at $u=0, 1, \cdots, n$.

Note that by definition, evaluating the polynomial at a specific positive integer is a WFOMC query.
However, such WFOMC is not readily solvable by existing algorithms (e.g., \cite{WFOMC-FO3, WFOMC-FO2-faster}) even for the \FOtwo{} fragment. In general, the runtime of these algorithms is exponential in the length of the sentence, while the length of the sentence $\sentence_{R,u}$ in the WFOMC query can be as large as the domain size $n$, hence these algorithms would not lead to polynomial runtime in the domain size.
In particular, the algorithm described in \Cref{sub:data-complexity} is not sufficient in this case either, as there would be $O(2^n)$ 1-types and, consequently, an exponential time data complexity.
To this end, we adapt this algorithm into a dynamic programming (DP) style algorithm that exploits the structure of $\sentence_{R,u}$ to achieve an efficient calculation of the polynomials for \FOtwo{} sentences, which then applies to \Ctwo{} sentences with cardinality constraints by techniques mentioned in \Cref{sub:data-complexity}.

% Let us first rewrite the computation~\eqref{eq:wfomc-fo2} into a DP.
% Define
% \begin{equation}
  % \begin{aligned}
    % g_{c_1, c_2, \dots, c_k} = \sum_{c_k = 0}^{n-c_1-\dots-c_{k-1}} &\binom{c_1+\dots+c_k}{c_k}\cdot g_{c_1,\dots, c_{k-1}}
    % \cdot s_k^{c_k(c_k-1)/2} z_k^{c_k}\prod_{1\le i < k} r_{i,k}^{c_i c_k},
  % \end{aligned}
  % \label{eq:wfomc-fo2-dp}
% \end{equation}
% and $g_{c_1} = s_1^{c_1(c_1-1)/2} z_1^{c_1}$.
% Then we can compute \eqref{eq:wfomc-fo2} by
% \begin{equation*}
  % \wfomc(\Psi, n, \weight, \negweight) = \sum_{|\config{c}| = n} g_{\config{c}}.
% \end{equation*}
% Simply applying the above DP algorithm to the \wcp{} will still lead to an exponential complexity, as the number of $g$'s is exponential in the number of 1-types.
% However, in the case of \wcp{}, we can exploit the structure of the 1-types to reduce the number of $g$'s to be computed.

Let $\Psi$ and $R$ be the input \FOtwo{} sentence and the binary relation that \wcp{} is defined on.
Consider the point evaluation of the \wcp{} at $u$.
Let $C_1(x), C_2(x), \cdots, \linebreak C_p(x)$ be the valid 1-types \footnote{We call a 1-type \emph{valid} if it is satisfiable in some model of the sentence.} of $\sentence$.
% Due to the auxiliary predicates $A_1, \cdots, A_u$ and $\bigwedge_{i=2}^u \forall x (A_i(x) \to A_{i-1}(x))$ in $\Psi_{R,u}$,
Due to the definition of $\sentence_{R,u}$, the valid 1-types of $\Psi_{R,u}$ must have the form $C_i(x) \wedge C^A_{i'}(x)$, where $C^A_{i'}(x)$ is defined as
$$
\begin{aligned}
C^A_{i'}(x) = A_1(x) \wedge A_2(x) \wedge \dots A_{i'}(x) \wedge \lnot A_{i'+1}(x) \wedge \dots \wedge \lnot A_u(x)
\end{aligned}
$$
for $i' \in [u]$, and
$$
C^A_{0}(x) = \lnot A_1(x) \wedge \lnot A_2(x) \wedge \dots \wedge \lnot A_u(x).
$$

We use the tuple $(i,i')$ to index the 1-type $C_i(x) \wedge C^A_{i'}(x)$.
Then we have the following observation.

\begin{observation}\label{observation:1}
%Reformulate the WFOMC of $\sentence_{R,u}$ to the WFOMC of a sentence $\forall x \forall y \ \phi(x,y)$ in prenex normal form with only universal quantifiers (removing existential quantifiers using the Skolemization technique~\cite{WFOMC-FO2}). Let $\phi_{(i,i'),(j,j')}(x,y)$ be the simplification of $\phi(x,y) \land \phi(y,x)$ when $x$ and $y$ belong to 1-type $(i,i')$ and $(j,j')$ respectively. Define
%$$
%r_{(i,i'),(j,j')} = \wmc(\phi_{(i,i'),(j,j')}(a,b), \weight, \negweight).
%$$

Let $r_{(i,i'),(j,j')}$ be the mutual 1-type coefficient (defined in \Cref{sub:data-complexity}) of $\sentence_{R,u}$.
It holds that
$$
r_{(i,i'),(j,j')} =
\begin{cases}
r_{i,j} = \wmc(\psi_{i,j}(a,b), \weight, \negweight),  & i' = j', \\
\begin{aligned}
    r^{\not-}_{i,j} = \wmc( & \psi_{i,j}(a,b) \wedge \lnot R(a,b)
     \wedge \lnot R(b,a), w, \overline{w}),
\end{aligned} & i' \neq j', \\
\end{cases}
$$
where $\psi_{i,j}(x,y)$ is the simplified formula of $\sentence$ defined in \Cref{sub:data-complexity}.
\end{observation}

This observation provides an important insight into the structure of the new 1-types (as shown in \Cref{fig:wcp}): we can imagine that the new 1-types of $\Psi_{R,u}$ form $u+1$ layers, with each layer containing a replica of the original 1-types of $\Psi$. If two elements $a, b$ fall in 1-types in the same layer, $R(a, b)$ can be either true or false. But if $a$ and $b$ are located in 1-types in different layers, $R(a, b)$ as well as $R(b, a)$ must be false due to the definition of $\Psi_{R,u}$ in Equation \eqref{eq:wcp}.

\begin{figure}
  \centering
  %\resizebox{0.5\textwidth}{!}{
  \begin{tikzpicture}[scale=0.6]
    \foreach \i in {0,...,\PICSIZE} {
        \draw [very thick,black] (\i,0) -- (\i,1);
    }
    \draw [very thick,black] (0,0) -- (\PICSIZE,0);
    \draw [very thick,black] (0,1) -- (\PICSIZE,1);
    \foreach \i in {0,...,3} {
        \draw [very thick,black] node [below] at (\i+0.5,0) {$\i$};
    }
    \draw [very thin,gray,step=0.25] (0,0) grid (\PICSIZE,1);
    \draw [very thick,black] node [below] at (5,0) {$\cdots$};
    \draw [very thick,black] node [below] at (7.5,0) {$u$};
    \draw [very thick,blue] (4.05,0.05) rectangle ++(0.9,0.9);
    \draw [very thick,red] (0.05,0.05) rectangle ++(3.9,0.9);
    \draw [<->, very thick, black] (1.5, 0.5) to [bend left=30] (4.4, 0.5) node [right] at (2.5, 1.5) {$r^{\not -}_{i,j}$};
    \draw [->, very thick, black] (5.5, 2) to (4.6, 0.5) node at (5.5, 2.2) {$r_{i,j}$};
  \end{tikzpicture}
  %}
  \caption{The 1-type structure of $\Psi_{R,u}$. Each tiny gray box represents an original 1-type of $\Psi$. The bold gray cells (each containing $4 \times 4$ boxes) refer to the replica of the original 1-types in each layer. The DP algorithm proceeds from left (cells $0, \cdots, 3$) to right (cell $4$).} \label{fig:wcp}
\end{figure}
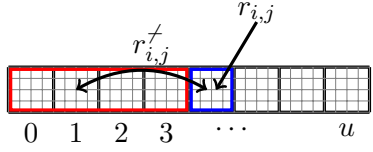

Therefore, we can perform a ``blocked'' DP, where each block consists of the 1-types in the same layer.
% Let us index the 1-types configuration by
We can write the new 1-type configuration as
$$c_{(1,1)}, \cdots, c_{(p,1)}, \  c_{(1,2)}, \cdots, c_{(p,2)}, \ \cdots,\  c_{(1,u)}, \cdots, c_{(p,u)}$$
where $c_{(i,i')}$ is the number of elements in the 1-type $(i,i')$.
% We define
% \begin{align*}
%   h_{\hat{u}, (c_1, c_2, \dots, c_p)} =
%   \sum_{c_{(1,1)}+c_{(2,1)}+\dots+c_{(\hat{u},1)} = c_1} \dots \sum_{c_{(1,L)}+c_{(2,L)}+\dots+c_{(\hat{u},L)} = c_p} g_{c_{1,1}, \dots, c_{\hat{u},1}, \dots, c_{1,L}, \dots, c_{\hat{u},L}}
% \end{align*}
% be the summation of the terms~\eqref{eq:wfomc-fo2-dp} over the layers $1, \dots, \hat{u}$ in the original DP algorithm.
For any $\hat{u}\in [u]$, we define $h_{\hat{u}, (c_1, c_2, \dots, c_p)}$ to be the weights of 1-types $(1,1), \cdots, (p,1), \cdots, (1,\hat{u}), \cdots, (p,\hat{u})$ such that $c_i = \sum_{j=1}^{\hat{u}} c_{(i,j)}$ for each $i \in [p]$, that is,
\begin{equation*}
\begin{aligned}
  h_{\hat{u}, (c_1, c_2, \dots, c_p)} &= \sum_{c_{(1,1)}+\cdots+c_{(1,\hat{u})} = c_1} \dots \sum_{c_{(p,1)}+\cdots+c_{(p,\hat{u})} = c_p} \frac{(\sum_{i \in [p]} c_i)!}{\prod_{i \in [p], i' \in [\hat{u}]} c_{(i,i')}!} \\
  &\prod_{i \in [p], i' \in [\hat{u}]} W(C_i, \weight, \negweight)^{c_{(i,i')}} \cdot s_i^{\binom{c_{(i,i')}}{2}} \prod_{\substack{i,j \in [p], i',j' \in [\hat{u}], \\ (i,i') < (j,j')}} \left(r_{(i,i'),(j,j')}\right)^{c_{(i,i')}c_{(j,j')}},
\end{aligned}
\end{equation*}
where $(i,i') < (j,j')$ means $i < j$ or $i = j \land i' < j'$.

For simplicity, we use $\config{c}$ to denote $(c_1, c_2, \dots, c_p)$.
% , and $\config{c}^*$ to denote $(c_{1,\hat{u}}, \dots, c_{L,\hat{u}})$.
Then the new term $h_{\hat{u}, \config{c}}$ can be also computed by a DP algorithm:
\begin{equation}
  \begin{aligned}
    h_{\hat{u}, \config{c}} = \sum_{\config{\bar c} + \config{c}^* = \config{c}} & \binom{|\config{c}|}{|\config{c}^*|} \cdot h_{\hat{u}-1, \config{\bar c}} \cdot w_{in}(\config{c}^*) \cdot w_{cross}(\config{\bar c}, \config{c}^*),
  \end{aligned}
  \label{eq:wfomc-fo2-dp-block}
\end{equation}
where recall that $\config{\bar c} + \config{c}^* = \config{c}$ means that $\bar c_i + c^*_i = c_i$ for each $i \in [p]$,
\begin{equation*}
w_{in}(\config{c}^*) = \binom{|\config{c}^*|}{\config{c}^*} \prod_{i=1}^p W(C_i, \weight, \negweight)^{c_i^*} \cdot s_i^{\binom{c_i^*}{2}} \prod_{1 \le i < j \le p} \left(r_{i,j}\right)^{c_i^* c_j^*}
\end{equation*}
is the weight within layer $\hat{u}$, and
\begin{equation*}
w_{cross}(\config{\bar c}, \config{c}^*) = \prod_{i=1}^p \prod_{j=1}^p \left(r_{i,j}^{\not-}\right)^{\bar{c}_i c_j^*}
\end{equation*}
is the weight between the layer $\hat{u}$ and the layers before $\hat{u}$.
The initial values are $h_{0, \config{c}} = w_{in}(\config{c})$ for each $\config{c}$.
Then the final WFOMC of $\Psi_{R,u}$ can be obtained by:
\begin{equation*}
  \wfomc(\sentence_{R,u}, n, \weight, \negweight) = \sum_{|\config{c}| = n} h_{u, \config{c}}.
\end{equation*}

% Therefore, we can assign elements into 1-types from layer $0$ to layer $u$ in a dynamic programming manner. Let $h_{u, \config{c}}$ be the total weights of element partitions into layers $0$ to $u$ with 1-type configuration $\config{c}$ (here we refer to the 1-type configuration of the original sentence $\sentence$), and suppose $\config{c} = \config{\bar c} + \config{c^*}$ where $\config{\bar c}$ is the 1-type configuration among layers $0, \dots, u-1$ and $\config{c^*}$ is the 1-type configuration of layer $u$. The value of each $h_{u, \config{c}}$ consists of the following parts:
% \begin{itemize}
%   \item $h_{u-1, \config{\bar c}}$, the weights among layers $0, \dots, u-1$,
%   \item $\binom{\config{c}}{\config{c^*}}$, the number of ways to choose $\config{c^*}$ from $\config{c}$ to be put in layer $u$,
%   \item $w_{in}(\config{c^*}) = \prod_{1 \le i < j \le p} r_{i,j}^{c^*_i \cdot c^*_j} \cdot \prod_{i=1}^p s_i^{c^*_i(c^*_i-1)/2} z_i^{c^*_i}$, the weight within layer $u$,
%   \item $w_{cross}(\config{\bar c}, \config{c^*}) = \prod_{i=1}^p \prod_{j=1}^p \left(r_{i,j}^0\right)^{\bar c_i \cdot c^*_j}$, the weight between 1-types before layer $u$ and 1-types in layer $u$ where $R(a,b)$ must be false.
% \end{itemize}

% Multiplying them together, we obtain
% $$
% h_{u, \config{c}} = \sum_{\config{\bar c} + \config{c^*} = \config{c}} h_{u-1, \config{\bar c}} \cdot \binom{\config{c}}{\config{c^*}} \cdot w_{in}(\config{c^*}) \cdot w_{cross}(\config{\bar c}, \config{c^*}).
% $$
The new DP algorithm is summarized in \Cref{wcp_dp}, where \texttt{interpolate\_1d($\dots$)} takes $n+1$ points as input and returns the corresponding polynomial of degree $n$.
%Note that we don't need to treat the evaluation of $u=0, 1, \cdots, n$ as separate tasks. One iteration of $u$ is sufficient to evaluate all of them.

\begin{algorithm}
\caption{Computing the \wcp{}} \label{wcp_dp}
\KwIn{A sentence $\sentence$, weighting functions $\weight, \negweight$, the relation $R$, and an integer $n$}
\KwOut{$f_n(u; \sentence, \weight, \negweight, R)$}
\SetKwFunction{intp}{interpolate\_1d}

% \lnote{I would like to make it explicit.}\\
$h_{0, \config{c}} \gets w_{in}(\config{c}) \textbf{ for each } |\config{c}| \le n$ \Comment{$\config{c} = (c_1, \cdots, c_p)$} \\
$evaluation\_pairs \gets \{(0, \sum_{|\config{c}|=n} h_{0, \config{c}})\}$ \\
\For{$u \gets 1$ \KwTo $n$} {
    \ForEach{\upshape $\config{c}$ \text{such that} $|\config{c}| \le n$} {
        $h_{u, \config{c}} \gets \sum_{\config{\bar c} + \config{c}^* = \config{c}} \binom{|\config{c}|}{|\config{c}^*|} \cdot h_{u-1, \config{\bar c}}  \cdot w_{in}(\config{c}^*) \cdot w_{cross}(\config{\bar c}, \config{c}^*)$
    }
    $evaluation\_pairs \gets evaluation\_pairs \cup \{(u, \sum_{|\config{c}|=n} h_{u, \config{c}})\}$
}
$f \gets$ \intp{$evaluation\_pairs$} \\
\Return $f$
\end{algorithm}

\begin{proposition}\label{prop:fast-wcp}
Computing \wcp{} for a given \Ctwo{} sentence with cardinality constraints, a domain of size $n$, a distinguished binary relation, and weighting functions can be completed in polynomial time in $n$.
\end{proposition}

\begin{proof}
  % Computing \wcp{} requires solving $n+1$ WFOMC problems of the form.
  As discussed in \Cref{sub:data-complexity}, any WFOMC problem for \Ctwo{} sentences with cardinality constraints can be reduced to WFOMC of an \FOtwo{} sentence, which can be solved in polynomial time in $n$ by the algorithm presented above.
  The interpolation step from $n+1$ point evaluations to the polynomial can also be done in polynomial time in $n$, completing the proof.
\end{proof}

\subsection{Applications} \label{sec:wcp-applications}

In this section, we show how \wcp{} can be used for computing WFOMC of \Ctwo{} sentences with cardinality constraints and several combinatorial axioms in polynomial time, i.e. we show that the fragments together with each of these axioms are domain-liftable. Some of these axioms were previously shown to be domain-liftable (e.g., connectedness axioms~\cite{WFOMC-axioms}, tree axiom~\cite{WFOMC-tree-axioms} and forest axiom~\cite{WFOMC-axioms}), while for others we show their domain-liftability for the first time (e.g., generalized connectedness axiom and bipartite axiom). We also show that \wcp{} allows us to efficiently compute the Tutte polynomial of any graph that can be encoded by a set of ground unary literals, cardinality constraints and a fixed \Ctwo{} sentence, e.g., block-structured graphs.

Most of the axioms in this section require some binary relation $R$ to be symmetric and irreflexive.
In this case, we append $\forall x \forall y \left( R(x,y) \to R(y,x) \right) \land \forall x \lnot R(x,x)$ to the input sentence $\sentence$, and interpret the graph $G(R)$ as a simple undirected graph.
% The new sentence of $\sentence$ with the additional symmetric formula is denoted by $\symsentence$.
We denote the new sentence by $\symsentence$.

\subsubsection{Axioms: Existing, Generalized, and New}

% We first show the ability of \wcp{} to handle some existing combinatorial axioms as well as their generalizations.
We first show how WFOMC with several existing and new combinatorial axioms can be computed using \wcp{}.

\noindentparagraph{$k$-connected-component axiom}
As in Proposition \ref{prop:wcp}, \wcp{} associates the weight of each model with a term indicating the number of connected components in $G(R)$, therefore we can restrict the number of connected components of $G(R)$.
The \emph{$k$-connected-component} axiom to the binary predicate $R$, denoted by $connected_k(R)$, requires $R$ to be symmetric and $G(R)$ to consist of exactly $k$ connected components. As a special case, $connected_1(R)$, which represents that $G(R)$ is an undirected connected graph, i.e., has just one connected component, was already shown to be domain-liftable in \cite{WFOMC-axioms}. Here we generalize their result to $connected_k(R)$.
\begin{theorem} \label{exmp:k-connected-component}
 The fragment of \Ctwo{} with cardinality constraints and a single $k$-connected-component axiom is domain-liftable.
\end{theorem}
\begin{proof}
    Consider a \Ctwo{} sentence $\sentence$ with cardinality constraints, a domain of size $n$, the weighting functions $\weight, \negweight$, and a distinguished binary predicate $R$. By \Cref{prop:wcp}, we have that $\wfomc(\sentence \land connected_k(R), n, \weight, \negweight)$ is equal to $ [u^k] f_n(u-1; \symsentence, \weight, \negweight, R)$.
\end{proof}

% We are enhancing it to a more generalized axiom.

\noindentparagraph{Bipartite axiom}
The \emph{bipartite} axiom to the binary predicate $R$, denoted by $bipartite(R)$, requires $R$ to be symmetric and $G(R)$ to be a bipartite graph\footnote{
Note that bipartite graphs are not finitely axiomatizable in first-order logic by the compactness theorem.}.
\begin{theorem} \label{exmp:bipartite}
    The fragment of \Ctwo{} with cardinality constraints and a single bipartite axiom is domain-liftable.
\end{theorem}
\begin{proof}
    Consider a \Ctwo{} sentence $\sentence$ with cardinality constraints, a domain of size $n$ weighting functions $\weight, \negweight$, and a distinguished binary predicate $R$.
    We need another two predicates to represent the bipartite graph. Define
    $$
    \begin{aligned}
    \sentence_b = &\symsentence \land \forall x \ ((P_1(x) \lor P_2(x)) \land (\lnot P_1(x) \lor \lnot P_2(x))) \\
    &\land \forall x \forall y \ \left( P_1(x) \land P_1(y) \lor P_2(x) \land P_2(y) \to \lnot R(x,y) \right),
    \end{aligned}
    $$
    where $\weight(P_1) = \negweight(P_1) = \weight(P_2) = \negweight(P_2) = 1$.
    The predicates $P_1, P_2$ divide domain elements into two groups to ensure that $G(R)$ is a bipartite graph. However, overcounting occurs since for every connected component in the bipartite graph, we have two choices of assigning vertices to $P_1$ and $P_2$.
    % (imagine that when doing 2-coloring on the component, the first vertex has two choices and the rest is determined by that choice),
    Hence, for each connected component, we need to multiply the WFOMC by a factor $\frac 12$ to its weight, which can be done using \wcp{}:
    \begin{align*}
      \wfomc(\sentence \land bipartite(R), n, \weight, \negweight)
      &= \sum_{\mu\in\fomodels{\sentence_b}{n}} W(\mu, \weight, \negweight)\cdot (\frac{1}{2})^{cc(\mu_R)} \\
      &= f_n(-\frac 12; \sentence_b, \weight, \negweight, R).
    \end{align*}
\end{proof}

\noindentparagraph{Tree and forest axioms}
The \emph{tree} axiom~\cite{WFOMC-tree-axioms} to the binary predicate $R$ (denoted by $tree(R)$) requires $R$ to be symmetric and $G(R)$ to be a tree.
The \emph{forest} axiom~\cite{WFOMC-axioms} to the binary predicate $R$ (denoted by $forest(R)$) requires $R$ to be symmetric and $G(R)$ to be a forest.
To handle these two axioms, we extend the ability of \wcp{} to classify the models of a sentence by both the number of connected components in $G(R)$ and the number of edges in $G(R)$ by introducing a new variable to the weight of $R$.

\begin{definition}
  \label{def:wcp-extended}
  With the notations in Definition \ref{def:wcp}, the extended \wcp{} of $\Psi$ is defined as a bivariate polynomial such that for all positive integers $u$ and all real numbers $v$,
  \begin{equation*}
    \hat f_n(u, v; \Psi, \weight, \negweight, R) = \wfomc(\Psi_{R,u}, n, \weight_{R,v}, \negweight),
  \end{equation*}
  where $w_{R,v}(R) = w(R) \cdot v$, and $w_{R,v}(P) = w(P)$ for any other predicate $P$.
\end{definition}

\begin{corollary}[Corollary from \Cref{prop:wcp,prop:fast-wcp}]\label{corol:wcp}
For the sentence $\Psi$ containing the \emph{symmetric} binary relation $R$, a domain of size $n$, and the weighting functions $\weight, \negweight$, we have
\begin{equation*}
\hat f_n(u, v; \Psi, \weight, \negweight, R) = \sum_{\mu \in \fomodels{\Psi}{n}} W(\mu, \weight, \negweight) \cdot (u+1)^{cc(\mu_R)} v^{2e(\mu_R)},
\end{equation*}
where $cc(\mu_R)$ is the number of connected components of $G(\mu_R)$ and $e(\mu_R)$ is the number of undirected edges\footnote{Here, as $R$ is symmetric, $G(\mu_R)$ is interpreted as an undirected graph, and the number of undirected edges is equal to half the cardinality of $R$ in $\mu$. It is where the number $2$ in the exponent comes from.}of $G(\mu_R)$. Furthermore, the extended WCP of $\sentence$ always exists, and can be computed in time polynomial in $n$ using the same techniques as in \Cref{prop:fast-wcp}.
\end{corollary}

With the above corollary, we can show that the tree and forest axioms are domain-liftable.
% Now we represent the tree axiom and the forest axiom using Corollary \ref{corol:wcp}.
% \begin{example}\label{tree-axiom}
\begin{theorem} \label{exmp:tree}
  The fragment of \Ctwo{} with cardinality constraints and a single tree axiom is domain-liftable.
\end{theorem}
\begin{proof}
  Consider a \Ctwo{} sentence $\sentence$ with cardinality constraints, a domain of size $n$, weighting functions $\weight, \negweight$ and a distinguished binary predicate $R$.
  By definition, the tree axiom requires that $G(R)$ should be connected and the number of edges in $G(R)$ should be $n-1$.
  By Corollary \ref{corol:wcp}, we have
  \begin{align*}
    \wfomc(\sentence\land &tree(R), n, \weight, \negweight)
    = [uv^{2(n-1)}] \hat f_{n}(u-1, v; \symsentence, \weight, \negweight, R).
  \end{align*}
\end{proof}
% \end{example}

\begin{theorem} \label{exmp:forest}
  The fragment of \Ctwo{} with cardinality constraints and a single forest axiom is domain-liftable.
\end{theorem}
\begin{proof}
  Consider a \Ctwo{} sentence $\sentence$ with cardinality constraints, a domain of size $n$, weighting functions $\weight, \negweight$ and a distinguished binary predicate $R$.
  By definition, the forest axiom requires the number of edges in $G(R)$ to be equal to the difference between the number of vertices in $G(R)$ and the number of connected components in $G(R)$.
  By Corollary \ref{corol:wcp}, we have
  \begin{align*}
    \wfomc(&\sentence \land forest(R), n, \weight, \negweight)
    = \sum_{i\in[n]} [u^iv^{2(n-i)}] \hat f_{n}(u-1, v; \symsentence, \weight, \negweight, R).
  \end{align*}
\end{proof}

\begin{remark}
  \label{remark:forest}
The tree axiom and the forest axiom have already been shown to be domain-liftable in \cite{WFOMC-tree-axioms} and \cite{WFOMC-axioms} respectively. Nevertheless, our method provides %a more intuitive view and
a more flexible command of the graph structures. For example, it is easy to add further constraints on the number of trees in the forest by our method (we will discuss this in \Cref{sec:axioms-combination}).
\end{remark}

% \subsubsection{New Axioms}

% \section{Additional Examples with the New Axioms}

% \subsubsection{Additional Examples with the New Axioms}

The strong expressive power of \Ctwo{} fragment with cardinality constraints and the axioms above allows us to represent many interesting combinatorial problems as well as Statistical Relational Learning (SRL) problems.
% that can be solved in polynomial time.
%Here we show one of them and refer the reader to Appendix~\ref{app:examples} for more examples.

\begin{example}\label{ex:permk}
  Consider a permutation $p_1, \dots, p_n$ over $[n]$, where $p_i$ is the position of $i$ in the permutation.
  The permutation can be regarded as a directed graph with several disjoint cycles if we add an edge from $i$ to $p_i$ for each $i \in [n]$.
  Counting the number of permutations with exactly $k$ cycles is a well-known combinatorial problem whose solution is given by the unsigned \emph{Stirling numbers of the first kind}.
  This problem can be readily solved by the $k$-connected-component axiom, where we allow the binary relation $R$ to be asymmetric.
  More precisely, we add the \emph{$k$-cycle-permutation} axiom to the binary predicate $R$, denoted by $perm_k(R)$, which requires $G(R)$ to be the graph corresponding to a permutation with exactly $k$ cycles. Given a \Ctwo{} sentence $\sentence$ with cardinality constraints, a domain of size $n$, weighting functions $\weight, \negweight$ and a distinguished binary predicate $R$, $\wfomc(\sentence \land perm_k(R), n, \weight, \negweight)$ is domain-liftable.
  In fact, let
  \begin{equation*}
    \sentence_p = \sentence \land \forall x \exists_{=1} y \ R(x,y) \land \forall x \exists_{=1} y \ R(y,x).
  \end{equation*}
  % Though $R$ may not be symmetric, to represent the constraint on the number of cycles we can define an auxiliary relation $R_s$, which is a symmetrized version of $R$.
  
  Then we have
  \begin{equation*}
    \wfomc(\sentence \land perm_k(R), n, \weight, \negweight) = [u^k] f_{n}(u-1; \sentence_p, \weight, \negweight, R).
  \end{equation*}
\end{example}

In the field of SRL, we can express \emph{hard} constraints on some relation by adding axioms to the relation, e.g., expressing a company hierarchy with the tree axiom, and player matching in a game with the bipartite axiom.
In the following example, we show that more complex \emph{soft} constraints can also be expressed with \wcp{}.

\begin{example}\label{exmp:expweight}
  In a social network, one may want to model the compactness in terms of a symmetric relation $friends$ by the number of connected components in $G(friends)$.
  In this case, a real number $d$ might be used to represent the degree of compactness: the larger $d$ means that the network is more compact, i.e., the number of connected components in $G(friends)$ is smaller.
  This can be achieved by defining the weight of a possible world $\omega$ as $W(\omega, \weight, \negweight)\cdot \exp(-d\cdot cc(\omega_{friends}))$.
  Then the reduced WFOMC of the inference problem is equal to $f_n(\exp(-d) - 1; \sentence, \weight, \negweight, friends)$, where $\sentence$ and $(\weight, \negweight)$ are the sentence and weighting functions modeling other aspects of the social network.
\end{example}

% \kqp{See if you want the following example.}

% \begin{example}
% When reducing MLN to WFOMC, we may consider the weighting function $\weight$ such that $\weight(R)$ is proportional or exponential to the number of weakly connected components of $G(R)$ (in which the WFOMC is not symmetric). Suppose we are given a \Ctwo{} sentence $\sentence$, a domain of size $n$, weighting functions $\weight, \negweight$ and a distinguished binary predicate $R$, where $\weight(R) = k \cdot cc(R)$. Then $\wfomc(\sentence, n, \weight, \negweight)$ is domain-liftable by computing
% $$
% \wfomc(\sentence, n, \weight, \negweight) = k \cdot f_n'(0; \sentence, \weight_1, \negweight, R),
% $$
% where $\weight_1 = \weight$ except that $\weight_1(R) = 1$, and $f_n'(u)$ is the derivative of $f_n(u)$.
% \end{example}

\subsubsection{From \wcp{} to Tutte Polynomial}

We can further use the extended \wcp{} in Definition \ref{def:wcp-extended} to obtain the well-known Tutte polynomial of undirected graphs that can be encoded by a set of ground unary literals, cardinality constraints and a fixed \Ctwo{} sentence.
Note that the size of the \Ctwo{} sentence should be fixed to a constant independent of the graph size (i.e., the number of vertices) while the parameters for cardinality constraints and the size of ground unary literals are unbounded.
The weighting functions are all set to $\One$ (i.e., a function always returns $1$) in the rest of this section.

\begin{theorem}\label{prop:tutte_wcp}
  Let $S$ be an infinite set of undirected graphs. Suppose that there is a fixed \Ctwo{} sentence $\sentence$ with a distinguished symmetric binary relation $E$ such that for each graph $G \in S$ of size $n$, there is a \Ctwo{} sentence $\sentence_G$ obtained by conjuncting $\sentence$ with cardinality constraints and ground unary literals, and $G$ is isomorphic to $G(\mu_E)$ for each model $\mu \in \fomodels{\sentence_G}{n}$.\footnote{For instance, \cite{C2graphs} gave the characterization of such graphs when the sentence is restricted to be in \Ctwo{}. Note that in their work the counting parameters and the size of the sentence are unbounded while in our case they should be constant to apply the fast calculation of \wcp{}.} Then the Tutte polynomial of any $G \in S$ of size $n$ can be computed as follows in time polynomial in $n$:
  \begin{equation*}
    T(x,y) = \frac{ \hat f_n \left( (x-1)(y-1)-1, \sqrt{y-1}; \sentence_T(G), \One, \One, R \right) }{ (x-1)^{cc(\mu_E)}(y-1)^n \wfomc(\sentence_G, n, \One, \One) },
  \end{equation*}
  where $\sentence_T(G) = \sentence_G \land \forall x\forall y \ \left( ( R(x, y) \to E(x, y) ) \land ( R(x, y) \to R(y, x) ) \right)$.
\end{theorem}

\begin{proof}
  Denote the edges of $G(\mu_E)$ by $\mathcal{E}$.
  By Corollary \ref{corol:wcp}, we have
  \begin{align*}
    & \hat f_n \left( (x-1)(y-1)-1, \sqrt{y-1}; \sentence_T(G), \One, \One, R \right) \\
    = & \sum_{\mu' \in \fomodels{\sentence_T(G)}{n}} (x-1)^{cc(\mu'_R)} (y-1)^{cc(\mu'_R) + e(\mu'_R)} \\
    = & \wfomc(\sentence_G, n, \One, \One) \cdot \sum_{A\subseteq \mathcal{E}} (x-1)^{cc(A)} (y-1)^{cc(A) + |A|} \\
    = & \wfomc(\sentence_G, n, \One, \One) \cdot (x-1)^{cc(\mu_E)}(y-1)^n \cdot T_G(x,y),
  \end{align*}
  where $T_G(x,y)$ is the Tutte polynomial of $G$ defined in \Cref{sec:graphpolys}.

  By \Cref{corol:wcp}, computing $f_n \left( (x-1)(y-1)-1, \sqrt{y-1}; \sentence_T(G), \One, \One, R \right)$ is in polynomial time in $n$. Since the cardinality of $S$ is infinite, the size of $\sentence$ is independent of the size of $G$, thus $\wfomc(\sentence_G, n, \One, \One)$ can be computed in time polynomial in $n$ by the algorithm described in \Cref{sub:data-complexity} as well. Therefore, we can obtain $T(x,y)$ in time polynomial in $n$.
\end{proof}

% \begin{remark}\label{prop:tutte-wcp-c2}
  % We can slightly relax the condition in \Cref{prop:tutte_wcp} to allow the sentence $\sentence$ to have multiple models for any $n$, as long as these models are isomorphic to $G_n$\footnote{For instance, \cite{C2graphs} gave the characterization of such graphs when the sentence is restricted to be in \Ctwo{}. Note that in their work the counting parameters and the size of the sentence are unbounded while in our case they should be constant to apply the fast calculation of \wcp{}.}.
  % Since the Tutte polynomial is invariant under graph isomorphism, we have
  % \begin{equation*}
    % T(x,y) = \frac{ f_n \left( (x-1)(y-1)-1, \sqrt{y-1}; \sentence_T, \One, \One, R \right) }{ (x-1)^{cc(\mu_E)}(y-1)^n \wfomc(\sentence, n, \One, \One) }.
  % \end{equation*}
% \end{remark}

\begin{example} \label{exmp:complete-graph-tutte-poly}
We can get the Tutte polynomial for a simple undirected complete graph of any size $n$. Let $S$ be the set of all simple undirected complete graphs and define
$$
\begin{aligned}
\sentence &= \forall x \lnot E(x,x), \\
\sentence_{G} &= \sentence \land (|E| = n(n-1))
\end{aligned}
$$
for each $G \in S$ of size $n$.

Then $\sentence_{G}$ has a unique model $\mu$ over the domain of size $n$ such that $G(\mu_E)$ is identical to $G$.
Therefore,
$$
\frac{ \hat f_n \left( (x-1)(y-1)-1, \sqrt{y-1}; \sentence_T(G), \One, \One, R \right) }{ (x-1)(y-1)^n }
$$
is the Tutte polynomial of the complete graph of size $n$, which can be computed in polynomial time in $n$ by Theorem \ref{prop:tutte_wcp}.
\end{example}

\begin{example}
  Consider the Tutte polynomial of a block-structured graph defined in Example \ref{ex:block-structured-graph}.
  We first fix the number $k$ of blocks and the connection among blocks. Let $S$ be the set of block-structured graphs of $n$ vertices such that each graph has $k$ blocks and the fixed connection. Define similarly to Equation \eqref{eq:block-structured-graph} the following sentences:
  $$
  \begin{aligned}
    \sentence = & \bigwedge_{V_i, V_j\text{ are connected}} (Block_i(x) \land Block_j(y) \to E(x,y)) \\
    & \land \bigwedge_{V_i, V_j\text{ are not connected}} (Block_i(x) \land Block_j(y) \to \lnot E(x,y)), \\
    \sentence_{G} = & \sentence \land \bigwedge_{e \in V} \left( Block_{v(e)}(e) \land \bigwedge_{j\neq v(e)} \lnot Block_{j}(e) \right)
  \end{aligned}
  $$
  for each $G \in S$ of size $n$.
  Then, the Tutte polynomial of $G$ is
  $$
  \frac{ \hat f_n \left( (x-1)(y-1)-1, \sqrt{y-1}; \sentence_T(G), \One, \One, R \right) }{ (x-1)^{cc(G)}(y-1)^n }.
  $$

  % Recall that by the technique in~\cite[Appendix A]{WFOMS-FO2-new} and \cite{WFOMC-C2}, we can reduce the calculation of the \wcp{} of $\sentence_B$ to the \wcp{} of an \FOtwo{} sentence, and thus obtain the following corollary.
  By Theorem \ref{prop:tutte_wcp}, the Tutte polynomial of $G$ can be computed in polynomial time in $n$. Applying the same proof to any constant $k$ and any possible connection among blocks, we conclude that the Tutte polynomial of any block-structured graph can be computed in polynomial time in the size of the graph.
\end{example}

The above result answers positively the question in \cite{WFOMC-tree-axioms} whether calculating the Tutte polynomial of a block-structured graph is hard. However, we remark that this result cannot be extended to computing the Tutte polynomial or evaluating any point for general graphs since representing a general graph by a first-order sentence involves a list of ground binary literals, making the WFOMC task \class{\#P}-hard \cite{WFOMC-conditioning}. Therefore we are consistent with the hardness of computing the Tutte polynomial of general graphs \cite{tuttepoly-hardness}, while showing tractability for block-structured ones (with a constant number of blocks). 

\section{\scpfull{}s}

The polynomial \wcp{} is insensitive to the orientation of the edges of $G(R)$, where $R$ is the distinguished relation. Hence, with \wcp{} we can only define axioms that do not depend on the orientation of the edges of the graph $G(R)$. In this section, we define two new polynomials, \emph{\sscpfull{}} (\sscp{}) and \emph{\nscpfull{}} (\nscp{}), that use the orientation of the edges as well. Both polynomials are collectively referred to as \emph{\scpfull{}s} (\scp{}s). We show that these new polynomials can be used to enforce several axioms, some of which have not been known to be domain-liftable before.

\subsection{Definition}

The definitions of \scp{}s are similar to \wcp{} except that we need another sequence of unary predicates to characterize the \emph{strong} connectedness in $G(R)$.

\begin{definition}[\scpfull{}s]\label{def:scp}
Let $\Psi$ be a first-order logic sentence possibly with cardinality constraints and ground unary literals, $w$ and $\overline{w}$ be two weighting functions and $R$ be a distinguished binary relation. The $n$-th \emph{\sscpfull{} (\sscp{})} $g_n(u,v)$ of $\sentence$ for the relation $R$ is the bivariate polynomial which satisfies:
$$
g_n(u,v;\Psi, w, \overline{w}, R) = \wfomc(\Psi_{R,u,v}, n, \weight, \negweight)
$$
for all positive integers $u,v$, where $\Psi_{R,u,v}$ is defined as:
\begin{align*}
    \Psi_{R,u,v} = \sentence \land & \bigwedge_{i=1}^u \forall x \forall y \left( A_i(x) \wedge (R(x,y) \vee R(y,x)) \to A_i(y) \right) \\
    \wedge & \bigwedge_{i = 2}^u \forall x \left( A_i(x) \to A_{i-1}(x) \right)\\
    \wedge & \bigwedge_{i=1}^{v-1} \forall x \forall y \left( B_i(x) \wedge R(x,y) \to B_{i+1}(y) \right) \\
    \wedge & \forall x \forall y \left( R(x,y) \to \lnot B_v(x) \land B_1(y) \right) \\
    \wedge & \bigwedge_{i = 2}^v \forall x \left( B_i(x) \to B_{i-1}(x) \right)
\end{align*}
where $A_i$'s and $B_i$'s are fresh unary predicates and $w(A_i) = \overline{w}(A_i) = w(B_i) = \overline{w}(B_i) = 1$.

The $n$-th \emph{\nscpfull{} (\nscp{})} $\bar g_n(u,v)$ of $\sentence$ for the relation $R$ is the bivariate polynomial which satisfies:
$$
\bar g_n(u,v;\Psi, w, \overline{w}, R) = \wfomc(\bar \Psi_{R,u,v}, n, \weight, \negweight)
$$
for all positive integers $u,v$, where $\bar \Psi_{R,u,v}$ is defined as:
\begin{align*}
    \bar \Psi_{R,u,v} = \sentence \land & \bigwedge_{i=1}^u \forall x \forall y \left( A_i(x) \wedge (R(x,y) \vee R(y,x)) \to A_i(y) \right) \\
    \wedge & \bigwedge_{i = 2}^u \forall x \left( A_i(x) \to A_{i-1}(x) \right)\\
    \wedge & \bigwedge_{i=1}^v \forall x \forall y \left( B_i(x) \wedge R(x,y) \to B_i(y) \right) \\
    \wedge & \bigwedge_{i = 2}^v \forall x \left( B_i(x) \to B_{i-1}(x) \right)
\end{align*}
where $A_i$'s and $B_i$'s are fresh unary predicates and $w(A_i) = \overline{w}(A_i) = w(B_i) = \overline{w}(B_i) = 1$.
\end{definition}

The predicates $A_1, \cdots, A_u$ treat $R$ symmetrically and are defined similarly as those in \wcp{}. As discussed in Proposition \ref{prop:wcp}, we can imagine them as capturing the information about weakly connected components of $G(R)$. The predicates $B_1, \cdots, B_v$ capture the directionality of edges in $G(R)$, enabling further directed axioms on $R$. We give an interpretation of both \sscp{} and \nscp{} by the following proposition and present more properties and applications in \Cref{sec:scp-applications}.

\begin{proposition}\label{prop:scp}
For the sentence $\Psi$, the weighting functions $w$ and $\overline{w}$ and the binary relation $R$, it holds that
\begin{equation*}
\begin{aligned}
g_n(u,v; \Psi, w, \overline{w}, R) &= \sum_{\mu \in \fomodels{\Psi}{n}} W(\mu, \weight, \negweight) \cdot (u+1)^{cc(\mu_R)} \cdot \chi_{G(\mu_R)}(v+1), \\
\bar g_n(u,v; \Psi, w, \overline{w}, R) &= \sum_{\mu \in \fomodels{\Psi}{n}} W(\mu, \weight, \negweight) \cdot (u+1)^{cc(\mu_R)} \cdot \bar \chi_{G(\mu_R)}(v+1),
\end{aligned}
\end{equation*}
where $cc(\mu_R)$ is the number of weekly connected components of $G(\mu_R)$, and $\chi_{G(\mu_R)}$ and $\bar \chi_{G(\mu_R)}$ are the strict and non-strict directed chromatic polynomials of $G(\mu_R)$ respectively.
\end{proposition}

\begin{proof}
We first prove the proposition for \sscp{}. Following the same argument in the proof of Proposition \ref{prop:wcp}, for any model $\mu'$ of $\sentence_{R,u,v}$, we can think of $\mu'$ as an extension of a model $\mu$ of $\sentence$ with the same weight, where the interpretation of $A_i$'s and $B_i$'s are determined by choosing labellings $i'\in[u]$ and $i''\in[v]$ for each domain element. Therefore, the polynomial $g_n(u,v; \Psi, w, \overline{w}, R)$ must be of the form $\sum_{\mu\in \fomodels{\Psi}{n}} W(\mu, \weight, \negweight) \cdot (u+1)^{cc(\mu_R)} \cdot q(v; \mu_R)$, where $q(v; \mu_R)$ is a term that depends on $v$ and $\mu_R$.
We can observe by Definition \ref{def:scp} that $q(v; \mu_R)$ is the number of vertex labellings of $G(\mu_R)$ by numbers $\{0,1,2,\cdots,v\}$ such that if there is an edge from vertex $x_1$ to vertex $x_2$, the label of $x_1$ is smaller than the label of $x_2$.
In other words, $q(v) = \chi_{G(\mu_R)}(v+1)$.

The same argument works for \nscp{} except that $q(v; \mu_R)$ is the number of vertex labellings of $G(\mu_R)$ by numbers $\{0,1,2,\cdots,v\}$ such that if there is an edge from vertex $x_1$ to vertex $x_2$, the label of $x_1$ is smaller than or equal to the label of $x_2$, which is the same meaning as $\bar \chi_{G(\mu_R)}(v+1)$.
\end{proof}

%Note that $q(v)$ will not change if we replace $G(\mu_R)$ by its condensation (i.e., a DAG in which vertices of each strongly connected component are shrunk to a single vertex). Thus, in what follows, we count the number of labellings of the condensation of $G(\mu_R)$, denoted by $ G_C$. We assume that $G_C$ contains $m$ vertices.

\begin{proposition}\label{prop:scp-degree}
$g_n(u,v;\Psi, w, \overline{w}, R)$ and $\bar g_n(u,v;\Psi, w, \overline{w}, R)$ always exist, and they are unique bivariate polynomials in which both variables have a degree at most $n$.
\end{proposition}

\begin{proof}
The existence of \sscp{} and \nscp{} follow immediately from Proposition \ref{prop:scp}. By Corollary \ref{corol:dichromatic-degree} and \ref{corol:nonstrictdichromatic-degree}, $\chi_{G(\mu_R)}(v)$ and $\bar \chi_{G(\mu_R)}(v)$ are polynomials of $v$ which has degree at most $n$. Since $cc(\mu_R) \le n$, the upper bounds of degree follow immediately from Proposition \ref{prop:scp}.
\end{proof}

\subsection{Fast Calculation} \label{sec:fastcalc-scp}

Similarly to \wcp{}, we can efficiently compute both \scp{}s for \Ctwo{} sentences with cardinality constraints.
We first elaborate the algorithm for \nscp{}, from which the algorithm for \sscp{} can be easily extended.

The algorithm, also based on DP, resembles the fast calculation of \wcp{} with the distinction that \wcp{} involves the 1-type structure of $u+1$ layers while here we are dealing with a $(u+1) \times (v+1)$ grid.
With the technique reducing WFOMC of \Ctwo{} with cardinality constriants to WFOMC of \FOtwo{} (refer to \Cref{sub:data-complexity}), in what follows, we assume that the input sentence is an \FOtwo{} sentence.

By Proposition \ref{prop:scp-degree}, $(n+1)^2$ point evaluations at $u=0, 1, \cdots, n$ and $v=0, 1, \cdots, n$ are sufficient to interpolate the polynomial (e.g., by 2D Lagrange Interpolation). Therefore, our task is to evaluate \nscp{} at these points.
Let $\Psi$ and $R$ be the \FOtwo{} sentence and the binary relation involved in the \nscp{}.
Consider the point evaluation of the \nscp{} at the point $(u, v)$.

% We index the 1-types by three indices.
Let $C_1(x), C_2(x), \cdots, C_L(x)$ be the 1-types of the sentence $\sentence$. We first define
$$
\begin{aligned}
C^A_{i'}(x) &= A_1(x) \wedge A_2(x) \wedge \dots A_{i'}(x) \wedge \lnot A_{i'+1}(x) \wedge \dots \wedge \lnot A_u(x)\\
C^B_{i''}(x) &= B_1(x) \wedge B_2(x) \wedge \dots B_{i''}(x) \wedge \lnot B_{i''+1}(x) \wedge \dots \wedge \lnot B_v(x),
\end{aligned}
$$
for $i' \in [u]$ and $i'' \in [v]$, and also
$$
\begin{aligned}
C^A_{0}(x) &= \lnot A_1(x) \wedge \lnot A_2(x) \wedge \dots \wedge \lnot A_u(x), \\
C^B_{0}(x) &= \lnot B_1(x) \wedge \lnot B_2(x) \wedge \dots \wedge \lnot B_v(x),
\end{aligned}
$$

Then the valid 1-types of $\sentence_{R,u,v}$ must have the form $C_i(x) \wedge C^A_{i'}(x) \wedge C^B_{i''}(x)$ and we index them by triples $(i,i',i'')$.
The following is a key observation for the fast calculation of $\wfomc(\Psi_{R,u,v}, n, \weight, \negweight)$.

\begin{observation}\label{observation:2}
%Reformulate the WFOMC of $\sentence_{R,u,v}$ to the WFOMC of the sentence $\forall x \forall y \ \varphi(x,y)$, where $\varphi(x,y)$ is a quantifier-free formula. Let $\varphi_{(i,i',i''),(j,j',j'')}(x,y)$ be the simplification of $\varphi(x,y) \land \varphi(y,x)$ when $x$ and $y$ belong to 1-type $(i,i',i'')$ and $(j,j',j'')$ respectively. Define
%$$
%r_{(i,i',i''),(j,j',j'')} = \wmc(\varphi_{(i,i',i''),(j,j',j'')}(a,b), \weight, \negweight).
%$$

Let $r_{(i,i',i''),(j,j',j'')}$ be the mutual 1-type coefficient (defined in \Cref{sub:data-complexity}) of $\sentence_{R,u,v}$.
It holds that
$$
\begin{aligned}
r_{(i,i',i''),(j,j',j'')} =
\begin{cases}
    r_{i,j} = \wmc(\psi_{i,j}(a,b), w, \overline{w}), & i' = j' \land i'' = j'', \\
    r^\to_{i,j} = \wmc(\psi_{i,j}(a,b) \wedge \lnot R(b,a), w, \overline{w}), & i' = j' \land i'' \neq j'', \\
    \begin{aligned}
        r^{\not-}_{i,j} = \wmc( & \psi_{i,j}(a,b) \wedge \lnot R(a,b) \\
        & \wedge \lnot R(b,a), w, \overline{w}),
    \end{aligned} & i' \neq j', \\
\end{cases}
\end{aligned}
$$
where $\psi_{i,j}(x,y)$ is the simplified sentence of $\sentence$ defined in \Cref{sub:data-complexity}.
\end{observation}

From the observation we can give an image of the structure of the new 1-types (as shown in \Cref{fig:scp}): we can imagine that $C^A_0(x), \cdots, C^A_u(x), C^B_0(x), \cdots, C^B_v(x)$ form a $(v+1) \times (u+1)$ grid, with each cell of the grid containing a replica of the original 1-types. If two domain elements $a, b$ fall in 1-types in the same cell, $R(a,b)$ can be either true or false. If $a, b$ are in the same column (i.e., in cells with the same $i'$) and w.l.o.g. $i'' < j''$, then only $R(a,b)$ is allowed to be true while $R(b,a)$ must be false. If $a, b$ are located in different columns (i.e., $i' \neq j'$), $R(a,b)$ and $R(b,a)$ must be false.

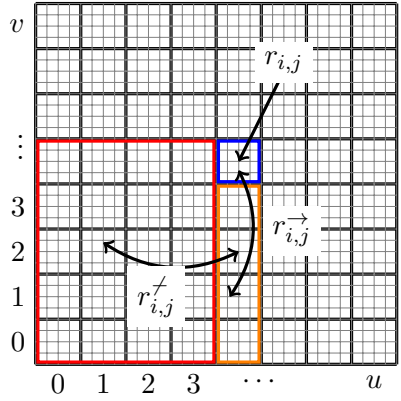
\begin{figure}
  \centering
  %\resizebox{0.5\textwidth}{!}{
  \begin{tikzpicture}[scale=0.6]
    \foreach \i in {0,...,\PICSIZE} {
        \draw [very thick,black] (\i,0) -- (\i,\PICSIZE);
        \draw [very thick,black] (0,\i) -- (\PICSIZE,\i);
    }
    \foreach \i in {0,...,3} {
        \draw [thin,black] node [below] at (\i+0.5,0) {$\i$};
        \draw [thin,black] node [left] at (0,\i+0.5) {$\i$};
    }
    \draw [very thin,gray,step=0.25] (0,0) grid (\PICSIZE,\PICSIZE);
    \draw [very thick,black] node [below] at (5,0) {$\cdots$};
    \draw [very thick,black] node [below] at (7.5,0) {$u$};
    \draw [very thick,black] node [left] at (0,5) {$\vdots$};
    \draw [very thick,black] node [left] at (0,7.5) {$v$};
    \draw [very thick,blue] (4.05,4.05) rectangle ++(0.9,0.9);
    \draw [very thick,orange] (4.05,0.05) rectangle ++(0.9,3.9);
    \draw [very thick,red] (0.05,0.05) rectangle ++(3.9,4.9);
    \draw [<->, very thick, black] (4.3, 1.5) to [bend right=30] (4.5, 4.3) node [right, fill=white] at (5, 2.95)  {$r^\to_{i,j}$};
    \draw [<->, very thick, black] (1.5, 2.7) to [bend right=30] (4.5, 2.5) node [right, fill=white] at (2, 1.5) {$r^{\not -}_{i,j}$};
    \draw [->, very thick, black] (5.5, 6.5) to (4.5, 4.5) node [fill=white] at (5.5, 6.7) {$r_{i,j}$};
  \end{tikzpicture}
  %}
  \caption{The 1-type structure of $\Psi_{R,u,v}$ which is a $(v+1) \times (u+1)$ grid. Each tiny gray box represents an original 1-type of $\Psi$. The bold gray cells (each containing $4 \times 4$ boxes) refer to the replicas of the original 1-types. The DP algorithm proceeds first along the vertical direction (from cell $(4,0)$ to cell $(4,4)$) and then along the horizontal direction (from columns $0, \cdots, 3$ to column $4$).} \label{fig:scp}
\end{figure}

By the above observation, we can efficiently compute the WFOMC of $\Psi_{R,u,v}$ by DP.
The DP proceeds in two directions in the grid: one is along the vertical direction of each $i'$-th column that calculates the weight of 1-types $(\cdot,i',0) \cdots, (\cdot, i', v)$, the other is along the horizontal direction that calculates the weight of 1-types $(\cdot, 0, \cdot) ,\cdots, (\cdot, u, \cdot)$. In both directions, we can perform a similar calculation as in the fast calculation of \wcp{}.
Denote by $c_{(i,i',i'')}$ the number of domain elements in the 1-type $(i,i',i'')$.
% by $\config{c}_{i',i''} = (c_{1, i', i''}, \cdots, c_{L, i', i''})$ a configuration of 1-types in the cell $(i', i'')$, and by $\config{c}_{i'} = (c_{1,i',0}, \cdots, c_{L,i',0}, \cdots, c_{1,i',v}, \cdots, c_{L,i',v})$ a configuration of 1-types in column $i'$.

Let us first consider the calculation along the vertical direction.
Note that all columns in the grid are identical in terms of the weight of configurations, thus we only need to compute an arbitrary column $i'\in[u]$.
We define $h^|_{i', \hat{v},\config{c}}$ as the summed weight of 1-types $(\cdot,i',0) \cdots, (\cdot, i', \hat{v})$ where $\config{c} = (c_1, \cdots, c_p)$ such that for each $i\in[p]$, $\sum_{i''=0}^{\hat{v}} c_{(i,i',i'')} = c_i$.
Then we have
$$
h^|_{i', \hat{v}, \config{c}} = \sum_{\config{\bar c} + \config{c^*} = \config{c}} h^|_{i', \hat{v}-1, \config{\bar c}} \cdot \binom{|\config{c}|}{|\config{c}^*|} \cdot w^|_{in}(\config{c^*}) \cdot w^|_{cross}(\config{\bar c}, \config{c^*}),
$$
where
\begin{equation}\label{eq:win-nscp}
w^|_{in}(\config{c^*}) = \binom{|\config{c}^*|}{\config{c}^*} \prod_{i=1}^p W(C_i,\weight,\negweight)^{c^*_i} \cdot s_i^{\binom{c^*_i}{2}} \prod_{1 \le i < j \le p} \left(r_{i,j}\right)^{c^*_i \cdot c^*_j}
\end{equation}
is the weight within cell $(i', \hat{v})$, and
$$w^|_{cross}(\config{\bar c}, \config{c^*}) = \prod_{i=1}^p \prod_{j=1}^p \left(r_{i,j}^\to\right)^{\bar c_i \cdot c^*_j}$$ is the weight between cells $(i', 0), \cdots, (i', \hat{v}-1)$ and cell $(i', \hat{v})$.
The initial value of $h^|_{i', 0, \config{c}}$ is $w_{in}^|(\config{c})$ for each $|\config{c}|\le n$.

In the same spirit, we can calculate the weight of the 1-type configurations in the horizontal direction.
We define $h^-_{\hat{u}, \config{c}}$ as the summed weight of 1-types $(\cdot, 0, \cdot) ,\cdots, (\cdot, \hat{u}, \cdot)$ where $\config{c} = (c_1, \cdots, c_p)$ such that for each $i\in[p]$, $\sum_{i'=0}^{\hat{u}} \sum_{i''=0}^v c_{(i,i',i'')} = c_i$.
By DP, the value of $h^-_{\hat{u}, \config{c}}$ can be calculated by
$$
h^-_{\hat{u}, \config{c}} = \sum_{\config{\bar c} + \config{c^*} = \config{c}} h^-_{\hat{u}-1, \config{\bar c}} \cdot \binom{|\config{c}|}{|\config{c^*}|} \cdot h^|_{\hat{u}, v, \config{c^*}} \cdot w^-_{cross}(\config{\bar c}, \config{c^*}),
$$
where $h^|_{\hat{u}, v, \config{c^*}}$, already calculated in the vertical direction, presents the weight within the column $\hat{u}$, and
$$w^-_{cross}(\config{\bar c}, \config{c^*}) = \prod_{i=1}^p \prod_{j=1}^p \left(r_{i,j}^{\not-}\right)^{\bar c_i \cdot c^*_j}$$
is the weight between columns $0, \cdots, \hat{u}-1$ and column $\hat{u}$.
One can easily check that the initial value of $h^-_{0, \config{c}}$ is $h^|_{0, v, \config{c}}$ for each $|\config{c}|\le n$.
Finally, the WFOMC of $\Psi_{R,u,v}$ can be obtained by summation over all possible configurations $\config{c}$:
\begin{equation*}
\wfomc(\sentence_{R,u,v},n,\weight,\negweight) = \sum_{|\config{c}| = n} h^-_{u, \config{c}}.
\end{equation*}

The overall calculation is summarized in \Cref{scp_dp}, where \texttt{interpolate\_2d($\dots$)} takes $(n+1)^2$ points as input and returns a bivariate polynomial with degree $n$ of each variable.
Since $h^|_{i', v, \config{c}}$ are the same for each $i'$, we omit the subscript $i'$ in the algorithm.
%Note that we don't need to treat the $(n+1)^2$ evaluations as separate tasks.

\begin{algorithm}
\caption{Computing the \scp{}s} \label{scp_dp}
\KwIn{A sentence $\sentence$, weighting functions $\weight, \negweight$, the relation $R$, and an integer $n$}
\KwOut{$g_n(u, v; \sentence, \weight, \negweight, R)$ if $w^|_{in}(\config{c^*})$ in \Cref{ln:nscp-or-sscp} is given by Equation \eqref{eq:win-nscp}, or $\bar g_n(u, v; \sentence, \weight, \negweight, R)$ if it is given by Equation \eqref{eq:win-sscp}}
\SetKwFunction{intp}{interpolate\_2d}

$h^|_{0, \config{c}} \gets w_{in}^|(\config{c}) \textbf{ for each } |\config{c}| \le n$ \\
$evaluation\_triples \gets \{(0, 0, \sum_{|\config{c}|=n} h^|_{0, \config{c}})\}$ \\
% $h^|_{-1, \bm 0} \gets 1$ \\
\For{$v \gets 1$ \KwTo $n$} {
    \ForEach{\upshape $\config{c}$ \text{such that} $|\config{c}| \le n$} {
        $h^|_{v, \config{c}} \gets \sum_{\config{\bar c} + \config{c^*} = \config{c}} h^|_{v-1, \config{\bar c}} \cdot \binom{|\config{c}|}{|\config{c}^*|} \cdot w^|_{in}(\config{c^*}) \cdot w^|_{cross}(\config{\bar c}, \config{c^*})$ \label{ln:nscp-or-sscp} \Comment{\footnotesize We obtain \nscp{} if $w^|_{in}(\config{c^*})$ is given by Equation \ref{eq:win-nscp}, or \sscp{} if it is given by Equation \ref{eq:win-sscp}} \\
    }
    % $h^-_{-1, \bm 0} \gets 1$ \\
    $h^-_{0, \config{c}} \gets h^|_{v, \config{c}} \textbf{ for each } |\config{c}| \le n$ \\
    $evaluation\_triples \gets evaluation\_triples \cup \{(0, v, \sum_{|\config{c}|=n} h^-_{0, \config{c}})\}$\\
    \For{$u \gets 1$ \KwTo $n$} {
        \ForEach{\upshape $\config{c}$ \text{such that} $|\config{c}| \le n$} {
            $h^-_{u, \config{c}} \gets \sum_{\config{\bar c} + \config{c^*} = \config{c}} h^-_{u-1, \config{\bar c}} \cdot \binom{|\config{c}|}{|\config{c^*}|} \cdot h^|_{u, \config{c^*}} \cdot w^-_{cross}(\config{\bar c}, \config{c^*})$
        }
        $evaluation\_triples \gets evaluation\_triples \cup \{(u, v, \sum_{|\config{c}|=n} h^-_{u, \config{c}})\}$
    }
}
$g \gets$ \intp{$evaluation\_triples$} \\
\Return $g$
\end{algorithm}

\begin{proposition}\label{prop:fast-nscp}
    Computing \nscp{} for a given \Ctwo{} sentence with cardinality constraints, a domain of size $n$, a distinguished binary relation, and weighting functions can be completed in polynomial time in $n$.
\end{proposition}

\begin{proof}
It follows from the discussion in this section and the same argument as in the proof of Proposition \ref{prop:fast-wcp}.
\end{proof}

\begin{proposition}\label{prop:fast-sscp}
Computing \sscp{} for a given \Ctwo{} sentence with cardinality constraints, a domain of size $n$, a distinguished binary relation, and weighting functions can be completed in polynomial time in $n$.
\end{proposition}

\begin{proof}
There are two differences between computing \sscp{} and computing \nscp{}. The first is that the definition of \sscp{} implies $\forall x \lnot R(x,x)$ implicitly. The second is that $R(a,b)$ and $R(b,a)$ should be false if $a,b$ are in the 1-types in the same cell, i.e., $r_{(i,i',i''),(j,j',j'')} = r^{\not-}_{i,j}$ when $i'=j' \land i''=j''$ and $s_i$ should be replaced with
\begin{equation*}
s^{\not-}_i = \wmc(\psi_i(a,b) \land \psi_i(b,a) \land \lnot R(a,b) \land \lnot R(b,a), \weight, \negweight).
\end{equation*}

Hence, a simple adaptation to Equation \eqref{eq:win-nscp} that removes 1-types containing $R(x,x)$ and replaces $r_{i,j}$ and $s_i$ with $r^{\not-}_{i,j}$ and $s^{\not-}_i$:
\begin{equation}\label{eq:win-sscp}
  w^|_{in}(\config{c^*}) = \binom{|\config{c}^*|}{\config{c}^*} \prod_{\substack{i=1, \\ \lnot R(x,x) \in C_i}}^p W(C_i,\weight,\negweight)^{c^*_i} \cdot \left(s^{\not-}_i\right)^{\binom{c^*_i}{2}} \prod_{\substack{1 \le i < j \le p, \\ \lnot R(x,x) \in C_i, \lnot R(x,x) \in C_j}} \left(r^{\not-}_{i,j}\right)^{c^*_i \cdot c^*_j}
\end{equation}
yields an algorithm for computing \sscp{}, as stated in \Cref{scp_dp}.
\end{proof}

\subsection{Applications}\label{sec:scp-applications}

Now, we show how \scp{} can be used for axioms that are associated with the orientation of the edges in $G(R)$, and for computing the directed chromatic polynomial of certain graphs.

\noindentparagraph{Strong connectedness axiom} The \emph{strong connectedness} axiom to the binary predicate $R$, denoted by $SC(R)$, requires $G(R)$ to be a strongly connected directed graph.

\begin{theorem}\label{exmp:strconnected}
    The fragment of \Ctwo{} with cardinality constraints and a single strong connectedness axiom is domain-liftable.
    More concretely, for any \Ctwo{} sentence $\sentence$ with cardinality constraints, a domain of size $n$, weighting functions $\weight, \negweight$ and a distinguished binary predicate $R$, it holds
    $$
    \begin{aligned}
        \wfomc(\sentence& \land SC(R), n, \weight, \negweight)
        = -[u] \bar g_{n}(u-1, -2; \sentence, \weight, \negweight, R).
        \end{aligned}
    $$
\end{theorem}

% We split the proof of \Cref{prop:strconnected} into several claims.

% \begin{claim}\label{claim:strconnected1}
% For a model $\mu$ of the sentence $\sentence_{R,u,v}$, if $G(\mu_R)$ is strongly connected then $q(v; \mu_R)$ in \Cref{prop:scp} equals to $v+1$.
% \end{claim}

% \begin{proof}
% The proof is similar to \Cref{prop:wcp}. Since $G(\mu_R)$ is strongly connected, the interpretation of $B_1(x), \cdots, B_v(x)$ should be the same for every domain element. Moreover, they have $v+1$ choices of the interpretation of $B_1(x), \cdots, B_v(x)$ regardless of interpretations of other predicates, so that $q(v; \mu_R) = v+1$.
% \end{proof}

% \begin{claim}\label{claim:strconnected2}
% For a model $\mu$ of the sentence $\sentence_{R,u,v}$, if $G(\mu_R)$ is weakly connected but not strongly connected then $q(v; \mu_R)$ in \Cref{prop:scp} is divisible by $v+2$.
% \end{claim}

% \begin{proof}
% Let $\chi_{G(\mu_R)}(v)$ and $\bar \chi_{G(\mu_R)}(v)$ be the strict and non-strict directed chromatic polynomial of $G(\mu_R)$ respectively. We know that $\bar \chi_{G(\mu_R)}(v) = q(v-1; \mu_R)$.
% When $G(\mu_R)$ is weakly connected but not strongly connected, it is impossible to color all vertices with a unique color such that if there is an edge from $x_1$ to $x_2$, the color of $x_1$ is strictly smaller than the color of $x_2$. Therefore, $\chi_{G(\mu_R)}(1) = 0$, and thus $q(-2; \mu_R) = \bar \chi_{G(\mu_R)}(-1) = (-1)^n \chi_{G(\mu_R)}(1) = 0$ by \Cref{lemma:dichromatic-strict}.
% \end{proof}

% \begin{proof}[Proof of \Cref{prop:strconnected}]
\begin{proof}
Recall from Proposition \ref{prop:scp} that $\bar g_n(u,v; \sentence, \weight, \negweight, R)$ equals to the sum of $W(\mu, \weight, \negweight) \cdot (u+1)^{cc(\mu_R)} \cdot \bar \chi_{G(\mu_R)}(v+1)$ over all $\mu \in \fomodels{\sentence}{n}$.
The term $[u]\bar g_{n}(u-1, v; \sentence, \weight, \negweight, R)$ only reserves the weight of models $\mu$ where $G(\mu_R)$ is weakly connected.

By Lemma \ref{lemma:dichromatic-negative} we know that $\bar \chi_{G(\mu_R)}(-1) = (-1) \cdot \chi_{acyc(G(\mu_R))}(1)$.
If $G(\mu_R)$ is strongly connected, $acyc(G(\mu_R))$ contains no edge, hence $\bar \chi_{G(\mu_R)}(-1) = -1$. If $G(\mu_R)$ is weakly connected but not strongly connected, $acyc(G(\mu_R))$ contains at least one edge so it is impossible to strictly color all vertices by one color. In this case, $\bar \chi_{G(\mu_R)}(-1) = 0$.

Therefore by Proposition \ref{prop:scp},
\begin{align*}
    -[u]\bar g_{n}(u-1, -2; \sentence, \weight, \negweight, R)
    =& -[u]\sum_{\mu \in \fomodels{\Psi}{n}} W(\mu, \weight, \negweight) \cdot u^{cc(\mu_R)} \cdot \bar \chi_{G(\mu_R)}(-1) \\
    =& -\sum_{\substack{\mu \in \fomodels{\sentence}{n}: \\ G(\mu_R)\text{ is strongly connected}}} -W(\mu, \weight, \negweight) \\
    =& \wfomc(\sentence \land SC(R), n, \weight, \negweight).
\end{align*}
\end{proof}

\begin{example}\label{exmp:sct}
A tournament is an orientation of an undirected complete graph. We can further require $G(R)$ to be strongly connected by adding the \emph{strongly connected tournament} axiom to the binary predicate $R$, denoted by $SCT(R)$. For any \Ctwo{} sentence $\sentence$ with cardinality constraints, a domain of size $n$, weighting functions $\weight, \negweight$ and a distinguished binary predicate $R$, the sentence
\begin{equation}\label{eq:tournament}
    \begin{aligned}
    \sentence_{TN} &= \sentence \land \forall x \left( \lnot R(x,x) \land Eq(x,x) \right) \land (|Eq| = n) \\
    &\land \forall x \forall y \big( \lnot Eq(x,y) \to ( R(x,y) \lor R(y,x) ) \land ( \lnot R(x,y) \lor \lnot R(y,x) ) \big)
    \end{aligned}
\end{equation}
restricts $G(R)$ to be a tournament. Then by \Cref{exmp:strconnected}, it holds that
\begin{align*}
    \wfomc(\sentence& \land SCT(R), n, \weight, \negweight)
    = -[u] \bar g_{n}(u-1, -2; \sentence_{TN}, \weight, \negweight, R).
\end{align*}

Therefore, the fragment of \Ctwo{} with cardinality constraints and a single strongly connected tournament axiom is domain-liftable.
\end{example}

\noindentparagraph{Acyclicity axiom} The \emph{acyclicity} axiom to the binary predicate $R$, denoted by $AC(R)$, requires that $G(R)$ does not have any directed cycle, or equivalently, $G(R)$ is a directed acyclic graph (DAG).

\begin{theorem}\label{exmp:acyclicity}
    The fragment of \Ctwo{} with cardinality constraints and a single acyclicity axiom is domain-liftable.
    More concretely, for any \Ctwo{} sentence $\sentence$ with cardinality constraints, a domain of size $n$, weighting functions $\weight, \negweight$ and a distinguished binary predicate $R$, it holds
    $$
    \begin{aligned}
        \wfomc(\sentence& \land AC(R), n, \weight, \negweight)
        = (-1)^n \cdot g_{n}(0, -2; \sentence, \weight, \negweight, R).
        \end{aligned}
    $$
\end{theorem}

\begin{proof}
By Lemma \ref{lemma:dichromatic-negative}, for each $\mu \in \fomodels{\sentence}{n}$ we have
\begin{equation*}
\chi_{G(\mu_R)}(-1) = \begin{cases}
    (-1)^n \bar \chi_{G(\mu_R)}(1) = (-1)^n, & G(\mu_R) \text{ is acyclic,} \\
    0, & \text{otherwise.}
\end{cases}
\end{equation*}

Therefore by Proposition \ref{prop:scp},
\begin{align*}
     (-1)^n \cdot g_{n}(0, -2; \sentence, \weight, \negweight, R)
    =& (-1)^n \sum_{\mu \in \fomodels{\Psi}{n}} W(\mu, \weight, \negweight) \cdot \chi_{G(\mu_R)}(-1) \\
    =& (-1)^n \sum_{\substack{\mu \in \fomodels{\sentence}{n}: \\ G(\mu_R)\text{ is acyclic}}} W(\mu, \weight, \negweight) \cdot (-1)^n \\
    =& \wfomc(\sentence \land AC(R), n, \weight, \negweight).
\end{align*}
\end{proof}

\noindentparagraph{Directed tree and forest axioms} The \emph{directed tree} axiom to the binary predicate $R$ and the unary predicate $Root$, denoted by $DT(R, Root)$, requires that $G(R)$ is a DAG with a single source, and the single source interprets $Root$ as true.
The \emph{directed forest} axiom to the binary predicate $R$, denoted by $DF(R)$, requires that $G(R)$ is a DAG such that every weakly connected component is a directed tree.
In other words, $G(R)$ is a DAG such that every vertex has at most one incoming edge.

\begin{theorem}\label{thm:dt-df}
    The fragment of \Ctwo{} with cardinality constraints and a single directed tree axiom or a single directed forest axiom is domain-liftable.
\end{theorem}
\begin{proof}
    The proof directly follows from \cite[Proposition 14 and 15]{WFOMC-axioms} that reduce WFOMC of \Ctwo{} with cardinality constraints and a single directed tree axiom or a single directed forest axiom to WFOMC of \Ctwo{} with cardinality constraints and a single acyclicity axiom.
\end{proof}

\noindentparagraph{Linear order axiom} The \emph{linear order} axiom to the binary predicate $R$, denoted by $LO(R)$, requires that $R$ represents a linear order. That is to say, $G(R)$ is an acyclic tournament with a self-loop for each vertex.
% We can force $G(R)$ to be a tournament by writing $\sentence_{TN}$ as \Cref{eq:tournament} and ensure its acyclicity by \Cref{prop:acyclicity}.
% but to ensure its acyclicity, we need to first obtain another property of the term $v^n$ in \scp{}.

\begin{theorem}\label{exmp:linearorder}
The fragment of \Ctwo{} with cardinality constraints and a single linear order axiom is domain-liftable.
\end{theorem}

% \begin{proof}
% For models $\mu \in \fomodels{\sentence}{n}$, only when $G(\mu_R)$ is acyclic can we color vertices using each of $n$ colors at least once, hence $\bar \chi^*_{G(\mu_R)}(n) \neq 0$ only if $G(\mu_R)$ is acyclic. Moreover, given that $G(\mu_R)$ is an acyclic tournament, $\bar \chi^*_{G(\mu_R)}(n) = 1$. By the definition of $\sentence_{TN}$ and \Cref{corol:nonstrictdichromatic-degree}, we have
% \begin{equation*}
% \begin{aligned}
% n! \cdot [uv^n]\bar g_n(u-1, v; \sentence_{TN}, \weight, \negweight, R)
% =& n! \cdot \sum_{\substack{\mu \in \fomodels{\sentence_{TN}}{n}: \\ G(\mu_R)\text{ is weakly connected}}} W(\mu, \weight, \negweight) \cdot \frac{\bar \chi^*_{G(\mu_R)}(n)}{n!} \\
% =& \sum_{\substack{\mu \in \fomodels{\sentence_{TN}}{n}: \\ G(\mu_R)\text{ is weakly connected} \\ \text{and acyclic}}} W(\mu, \weight, \negweight) \cdot \bar \chi^*_{G(\mu_R)}(n) \\
% =& \sum_{\substack{\mu \in \fomodels{\sentence}{n}: \\ G(\mu_R)\text{ is an acyclic tournament}}} W(\mu, \weight, \negweight) \\
% =& \wfomc(\sentence \land LO(R), n, \weight, \negweight).
% \end{aligned}
% \end{equation*}
% \end{proof}

\begin{proof}
Let $R'$ be a fresh binary relation with $\weight(R') = \negweight(R') = 1$ and define $\sentence'_{TN}$ as Equation \eqref{eq:tournament} but on relation $R'$. Define $\sentence_{LO}$ as
\begin{equation*}
\sentence_{LO} = \sentence'_{TN} \land (\forall x \ R(x,x)) \land (\forall x \forall y (\lnot Eq(x,y) \to (R(x,y) \leftrightarrow R'(x,y)))).
\end{equation*}

Then \Cref{exmp:acyclicity} yields the following equation:
\begin{equation*}
  (-1)^n \cdot g_{n}(0, -2; \sentence_{LO}, \weight, \negweight, R') = \wfomc(\sentence \land LO(R), n, \weight, \negweight).
\end{equation*}
\end{proof}

\begin{remark}
    \label{remark:k-acyclicity}
Though some of the axioms have been proven to be domain-liftable such as the acyclicity axiom in \cite{WFOMC-axioms} and the linear order axiom in \cite{WFOMC-linearorder-axiom}, by our approach we can further generalize these axioms to obtain more features such as restricting the number of weakly connected components of $G(R)$ of an acyclic graph. For example, if we require $G(R)$ to be the union of $k$ disjoint directed acyclic graphs, we can obtain the WFOMC by $(-1)^n \cdot [u^k] g_n(u-1,-2; \sentence, \weight, \negweight, R)$.
\end{remark}

\noindentparagraph{From \sscp{} and \nscp{} to Directed Chromatic Polynomials}
With the same idea of obtaining the Tutte polynomial for graphs encoded by a set of ground unary literals, cardinality constraints and a fixed \Ctwo{} sentence, we can recover the strict and non-strict directed chromatic polynomials from \sscp{} and \nscp{} respectively of such sentences encoding directed graphs.

\begin{corollary}[Corollary from \Cref{prop:tutte_wcp}]\label{prop:dichromatic_scp}
Let $S$ be an infinite set of directed graphs. Suppose that there is a fixed \Ctwo{} sentence $\sentence$ with a distinguished binary relation $E$ such that for each $D \in S$ of size $n$, there is a \Ctwo{} sentence $\sentence_D$ obtained by conjuncting $\sentence$ with cardinality constraints and ground unary literals, and $D$ is isomorphic to $G(\mu_E)$ for each model $\mu \in \fomodels{\sentence_D}{n}$. Then the strict and non-strict directed chromatic polynomials of any $D \in S$ of size $n$ can be computed as follows in time polynomial in $n$:
  \begin{equation*}
  \begin{aligned}
    \chi_D(x) &= \frac{ g_n \left(0, x; \sentence_D, \One, \One, E \right) }{ \wfomc(\sentence_D, n, \One, \One) }, \\
    \bar \chi_D(x) &= \frac{ \bar g_n \left(0, x; \sentence_D, \One, \One, E \right) }{ \wfomc(\sentence_D, n, \One, \One) },
  \end{aligned}
  \end{equation*}
where $\One$ is the weighting function that always returns 1.
\end{corollary}

\section{Combinations of Axioms}
\label{sec:axioms-combination}

An important advantage of the proposed approach over the prior works is that it can easily extend to new axioms by combining existing axioms.
For instance, consider the \textit{polytree} axiom that requires the interpretation of a distinguished binary predicate to be a polytree, i.e., a DAG whose underlying undirected graph is a tree.
Adopting the methodology of existing research~\cite{WFOMC-tree-axioms,WFOMC-linearorder-axiom,WFOMC-axioms}, one would typically delve into the literature on polytree enumeration \cite{polytree} and attempt to adapt the techniques (if they exist) to the WFOMC problem.
However, even if such techniques are available, integrating them into current WFOMC algorithms is not a trivial task.
In contrast, our approach leverages a single polynomial to encapsulate the structural characteristics of the models.
When computing WFOMC with an axiom, the process involves extracting the necessary information from the polynomials, either by algebraically manipulating the coefficients or by evaluating the polynomial at specific points.
Thus, to support a new axiom that is expressible as a combination of existing axioms, we can directly and synthetically apply the existing retrieval techniques.
This streamlined approach eliminates the need for further algorithmic development for those axioms that can be captured by one of the polynomials.

Indeed, we have demonstrated examples of axiom combinations in Remarks \ref{remark:forest} and \ref{remark:k-acyclicity}, where the WFOMC problems present substantial challenges that existing methods have not yet addressed.
The domain-liftability of these new axioms can be easily demonstrated through the fast calculation of their corresponding polynomials.

\begin{theorem}\label{thm:combined-axioms}
    The fragment of \Ctwo{} with cardinality constraints and a single axiom in the following list is domain-liftable.
    \begin{itemize}
      \item $bipartite_k(R) = connected_k(R) \land bipartite(R)$, requiring $G(R)$ to be a bipartite graph with $k$ connected components;
      \item $forest_k(R) = connected_k(R) \land forest(R)$, requiring $G(R)$ to be a forest with $k$ trees;
      \item $AC_k(R) = connected_k(R) \land AC(R)$, requiring $G(R)$ to be a DAG with $k$ weakly connected components;
      \item $BiAC(R) = bipartite(R) \land AC(R)$, requiring $G(R)$ to be a DAG whose underlying undirected graph is bipartite;
      \item $polytree(R) = tree(R) \land AC(R)$, requiring $G(R)$ to be a polytree, i.e., a DAG whose underlying undirected graph is a tree;
      \item $polyforest(R) = forest(R) \land AC(R)$, requiring $G(R)$ to be a polyforest, i.e., a DAG whose underlying undirected graph is a forest.
    \end{itemize}

    Note that in the combination of axioms involving $AC(R)$, we do not require $R$ to be symmetric.
\end{theorem}
\begin{proof}
    Consider a \Ctwo{} sentence $\sentence$ with cardinality constraints, a domain of size $n$, weighting functions $\weight, \negweight$ and a distinguished binary predicate $R$.
    We prove the theorem by showing that the WFOMC of the sentence with each of the axioms can be derived from \wcp{} or \sscp{}.
    Let $\sentence_R = \sentence \land \forall x\forall y (R(x,y)\to R(y,x))\land \forall \lnot R(x,x)$,
    \begin{align*}
        \sentence_b = &\sentence \land \forall x \ ((P_1(x) \lor P_2(x)) \land (\lnot P_1(x) \lor \lnot P_2(x))) \\
        &\land \forall x \forall y \ \left( P_1(x) \land P_1(y) \lor P_2(x) \land P_2(y) \to \lnot R(x,y) \land \lnot R(y,x) \right),
    \end{align*}
    and let $\hat f_n$ be the extended \wcp{} defined in Definition \ref{def:wcp-extended}.
    Then we have
    \begin{align*}
        \wfomc(\sentence \land bipartite_k(R), n, \weight, \negweight) &= (\frac{1}{2})^k \cdot [u^k]f_n(u-1; \sentence_R, \weight, \negweight, R), \\
        \wfomc(\sentence \land forest_k(R), n, \weight, \negweight) &= [u^kv^{2(n-k)}] \hat f_n(u-1, v; \sentence_R, \weight, \negweight, R), \\
        \wfomc(\sentence \land AC_k(R), n, \weight, \negweight) &= (-1)^n \cdot [u^k]g_n(u-1, -2; \sentence, \weight, \negweight, R), \\
        \wfomc(\sentence \land BiAC(R), n, \weight, \negweight) &= (-1)^n \cdot g_n(-\frac{1}{2}, -2; \sentence_b, \weight, \negweight, R), \\
        \wfomc(\sentence \land polytree(R), n, \weight, \negweight) &= (-1)^n \cdot [uv^{2(n-1)}]g_n(u-1, -2, v; \sentence, \weight, \negweight, R), \\
        \wfomc(\sentence \land polyforest(R), n, \weight, \negweight) &= (-1)^n \cdot \sum_{i\in[n]} [u^iv^{2(n-i)}] g_n(u-1, -2, v; \sentence, \weight, \negweight, R). \tag*{\qedhere}
    \end{align*}
\end{proof}

\section{Numerical Results}

We implemented Algorithm \ref{wcp_dp} and Algorithm~\ref{scp_dp}\footnote{The source code can be found in \url{https://github.com/l2l7l9p/Polynomials-for-WFOMC}.} and used them on several classes of combinatorial structures shown in \Cref{tab:axioms} and checked the results against the Online Encyclopedia of Integer Sequences (OEIS).\footnote{\url{https://oeis.org/}}
In the evaluation, we counted the structures described below, where we use
\begin{equation*}
  \sentence_{ug} = \forall x \lnot R(x,x) \land \forall x \forall y \left( R(x,y) \to R(y,x) \right)
\end{equation*}
to denote the sentence encoding simple undirected graphs and
\begin{equation*}
  \sentence_{dg} = \forall x \lnot R(x,x)
\end{equation*}
to denote the sentence encoding simple directed graphs. We use the function $\One$ that always returns 1 as the weighting functions.

%Counting over $n$ labeled nodes the certain combinatorial structures that we convert to the axioms in Section \ref{sec:wcp-applications} and Section \ref{sec:scp-applications} is also a well-studied problem. Several counting results are presented on the Online Encyclopedia of Integer Sequences (OEIS) \footnote{\url{https://oeis.org/}}. Therefore, we can verify the correctness of our approaches for those axioms by matching the results with the corresponding integer sequence on OEIS.
%For example, the number of connected labeled undirected graphs with $n$ nodes
%which corresponds to OEIS sequence A001187.
%If we compute $\wfomc(\sentence_{ug} \land connected_1(R), n, \weight, \negweight)$ where $\weight(P)=\negweight(P)=1$ for each predicate $P$ and $\sentence_{ug}$ encodes an arbitrary undirected graph defined in Appendix \ref{sec:performance}, the result ought to match A001187.

%We test the following axioms by matching them to the corresponding integer sequences. Here we use the weighting functions $\weight, \negweight$ such that $\weight(P)=\negweight(P)=1$ for each predicate $P$. For the axioms presented in Section \ref{sec:wcp-applications}, we use the sentence $\sentence_{ug}$ defined in Appendix \ref{sec:performance}. For the axioms presented in Section \ref{sec:scp-applications}, we use the sentence $\sentence_{dg} = \forall x \lnot R(x,x)$ which encodes an arbitrary directed graph.

\begin{table}[t]
  \centering
  \caption{Computational results of undirected axioms}\label{tab:oeis-undirected}
  %\begin{tabular}{c|p{5.5cm}|l}
  \begin{tabular}{c|p{0.8\textwidth}|l|l}
    % connectedness
    \hline
    $n$ & $f_n(u-1; \sentence_{ug}, \One, \One, R)$ & \multicolumn{2}{l}{$occ_n$} \\
    \hline
    1 & $u$ & \multicolumn{2}{l}{1} \\
    2 & $u + u^2$ & \multicolumn{2}{l}{1} \\
    3 & $4 u + 3 u^2 + u^3$ & \multicolumn{2}{l}{4} \\
    4 & $38 u + 19 u^2 + 6 u^3 + u^4$ & \multicolumn{2}{l}{38} \\
    5 & $728 u + 230 u^2 + 55 u^3 + 10 u^4 + u^5$ & \multicolumn{2}{l}{728} \\
    6 & $26704 u + 5098 u^2 + 825 u^3 + 125 u^4 + 15 u^5 + u^6$ & \multicolumn{2}{l}{26704} \\
    ... & ... & \multicolumn{2}{l}{...} \\
    % bipartite
    \hline\hline
    $n$ & $f_n(u; (\sentence_{ug})_b, \One, \One, R)$ & \multicolumn{2}{l}{$b_n$} \\
    \hline
    1 & $2 + 2 u$ & \multicolumn{2}{l}{1} \\
    2 & $6 + 10 u + 4 u^2$ & \multicolumn{2}{l}{2} \\
    3 & $26 + 54 u + 36 u^2 + 8 u^3$ & \multicolumn{2}{l}{7} \\
    4 & $162 + 366 u + 300 u^2 + 112 u^3 + 16 u^4$ & \multicolumn{2}{l}{41} \\
    5 & $1442 + 3270 u + 2860 u^2 + 1320 u^3 + 320 u^4 + 32 u^5$ & \multicolumn{2}{l}{376} \\
    6 & $18306 + 39446 u + 33540 u^2 + 16640 u^3 + 5040 u^4 + 864 u^5 + 64 u^6$ & \multicolumn{2}{l}{5177} \\
    ... & ... & \multicolumn{2}{l}{...} \\
    % tree and forest
    \hline\hline
    $n$ & $f_n(u-1, v; \sentence_{ug}, \One, \One, R)$ & $tr_n$ & $fo_n$ \\
    \hline
    1 & $u$ & 1 & 1 \\
    2 & $u v^2 + u^2$ & 1 & 2 \\
    3 & $3 u v^4 + u v^6 + 3 u^2 v^2 + u^3$ & 3 & 7 \\
    4 & $16 u v^6 + 15 u v^8 + 6 u v^{10} + u v^{12} + 15 u^2 v^4 + 4 u^2 v^6 + 6 u^3 v^2 + u^4$ & 16 & 38 \\
    5 & $125 u v^8 + 222 u v^{10} + 205 u v^{12} + 120 u v^{14} + 45 u v^{16} + 10 u v^{18} + u v^{20} + 110 u^2 v^6 + 85 u^2 v^8 + 30 u^2 v^{10} + 5 u^2 v^{12} + 45 u^3 v^4 + 10 u^3 v^6 + 10 u^4 v^2 + u^5$ & 125 & 291 \\
    6 & $1296 u v^{10} + 3660 u v^{12} + 5700 u v^{14} + 6165 u v^{16} + 4945 u v^{18} + 2997 u v^{20} + 1365 u v^{22} + 455 u v^{24} + 105 u v^{26} + 15 u v^{28} + u v^{30} + 1080 u^2 v^8 + 1617 u^2 v^{10} + 1330 u^2 v^{12} + 735 u^2 v^{14} + 270 u^2 v^{16} + 60 u^2 v^{18} + 6 u^2 v^{20} + 435 u^3 v^6 + 285 u^3 v^8 + 90 u^3 v^{10} + 15 u^3 v^{12} + 105 u^4 v^4 + 20 u^4 v^6 + 15 u^5 v^2 + u^6$ & 1296 & 2932 \\
    ... & ... & ... & ...
  \end{tabular}
\end{table}

\begin{description}
  \item[Connectedness axiom] Both OEIS sequence A001187 and $\wfomc(\sentence_{ug} \land connected_1(R), n, \One, \One)$ count the number of connected undirected graphs on $n$ labeled nodes. By \Cref{exmp:k-connected-component}, we compute $occ_n = [u]f_n(u-1; \sentence_{ug}, \One, \One, R)$ for each positive integer $n$ shown in \Cref{tab:oeis-undirected}. The sequence $occ_n$ matches the sequence A001187.
      % \begin{itemize}
      %   \item $f_n(u) = 1 + u$, and $[u]f_n(u-1) = 1$;
      %   \item $f_n(u) = 2 + 3 u + u^2$, and $[u]f_n(u-1) = 1$;
      %   \item $f_n(u) = 8 + 13 u + 6 u^2 + u^3$, and $[u]f_n(u-1) = 4$;
      %   \item $f_n(u) = 64 + 98 u + 43 u^2 + 10 u^3 + u^4$, and $[u]f_n(u-1) = 38$;
      %   \item $f_n(u) = 1024 + 1398 u + 465 u^2 + 105 u^3 + 15 u^4 + u^5$, and $[u]f_n(u-1) = 728$;
      %   \item $f_n(u) = 32768 + 39956 u + 8488 u^2 + 1495 u^3 + 215 u^4 + 21 u^5 + u^6$, and $[u]f_n(u-1) = 26704$;
      %   \item ...
      % \end{itemize}

  \item[Bipartite axiom] Both OEIS sequence A047864 and $\wfomc(\sentence_{ug} \land bipartite(R), n, \One, \One)$ count the number of bipartite graphs on $n$ labeled nodes. By \Cref{exmp:bipartite}, we compute $b_n = f_n(-\frac 12; (\sentence_{ug})_b, \One, \One, R)$ for each positive integer $n$ shown in \Cref{tab:oeis-undirected}. The sequence $b_n$ matches the sequence A047864.

  \item[Tree axiom] Both OEIS sequence A000272 and $\wfomc(\sentence_{ug} \land tree(R), n, \One, \One)$ count the number of trees on $n$ labeled nodes. By \Cref{exmp:tree}, we compute $tr_n = [uv^{2(n-1)}]\hat f_n(u-1, v; \sentence_{ug}, \One, \One, R)$ for each positive integer $n$ shown in \Cref{tab:oeis-undirected}. The sequence $tr_n$ matches the sequence A000272.

  \item[Forest axiom] Both OEIS sequence A001858 and $\wfomc(\sentence_{ug} \land forest(R), n, \One, \One)$ count the number of forests on $n$ labeled nodes. By \Cref{exmp:forest}, we compute $fo_n = \sum_{i\in[n]} [u^iv^{2(n-i)}] \hat f_{n}(u-1, v; \sentence_{ug}, \One, \One, R)$   for each positive integer $n$ shown in \Cref{tab:oeis-undirected}. The sequence $fo_n$ matches the sequence A001858.

  \item[Strong connectedness axiom] Both OEIS sequence A003030 and $\wfomc(\sentence_{dg} \land SC(R), n, \One, \One)$ count the number of strongly connected directed graphs on $n$ labeled nodes. By \Cref{exmp:strconnected}, we compute $sc_n = -[u]\bar g_n(u-1, -2; \sentence_{dg}, \One, \One, R)$ for each positive integer $n$ shown in \Cref{tab:oeis-directed}. The sequence $sc_n$ matches the sequence A003030.

  \item[Strongly connected tournament axiom] Both OEIS sequence A054946 and $\wfomc(\sentence_{dg} \land SCT(R), n, \One, \One)$ count the number of strongly connected tournaments on $n$ labeled nodes. By Example \ref{exmp:sct}, we compute $sct_n = -[u]g_n(u-1, -2; (\sentence_{dg})_{TN}, \One, \One, R)$ for each positive integer $n$ shown in Table \Cref{tab:oeis-directed}. The sequence $sct_n$ matches the sequence A054946.

  \item[Acyclicity axiom] Both OEIS sequence A003024 and $\wfomc(\sentence_{dg} \land AC(R), n, \One, \One)$ count the number of directed acyclic graphs on $n$ labeled nodes. By \Cref{exmp:acyclicity}, we compute $ac_n = (-1)^n \cdot g_n(0, -2; \sentence_{dg}, \One, \One, R)$ for each positive integer $n$ shown in \Cref{tab:oeis-directed}. The sequence $ac_n$ matches the sequence A003024.

  \item[Connected acyclicity axiom] Both OEIS sequence A082402 and $\wfomc(\sentence_{dg} \land AC_1(R), n, \One, \One)$ count the number of weakly connected directed acyclic graphs on $n$ labeled nodes. By \Cref{thm:combined-axioms}, we compute $cac_n = (-1)^n \cdot [u]g_n(u-1, -2; \sentence_{dg}, \One, \One, R)$ for each positive integer $n$ shown in \Cref{tab:oeis-directed}. The sequence $ac_n$ matches the sequence A082402.

  \item[Directed tree axiom] Both OEIS sequence A000169 and $\wfomc(\sentence_{dt} \land AC(R), n, \One, \One)$ count the number of directed trees on $n$ labeled nodes, where
      \begin{equation*}
      \sentence_{dt} = \forall x \left( \lnot root(x) \to \exists_{=1} y E(y,x) \right) \land |root| = 1.
      \end{equation*}
      By \Cref{thm:dt-df} and \cite[Proposition 14]{WFOMC-axioms}, we compute $dt_n = (-1)^n \cdot g_n(0, -2; \sentence_{dt}, \One, \One, R)$ for each positive integer $n$ shown in \Cref{tab:oeis-directed}. The sequence $dt_n$ matches the sequence A000169.
\end{description}

\begin{center}
  \begin{longtable}{c|p{0.7\textwidth}|l|l}
  \caption{Computational results of directed axioms}\label{tab:oeis-directed} \\
    % strong connectedness
    \hline
    $n$ & $\bar g_n(u-1, v; \sentence_{dg}, \One, \One, R)$ & \multicolumn{2}{l}{$sc_n$} \\
    \hline
    1 & $u + u v$ & \multicolumn{2}{l}{1} \\
    2 & $3 u + 4 u v + u v^2 + 2 u^2 v + u^2 v^2 + u^2$ & \multicolumn{2}{l}{1} \\
    3 & $54 u + 80 u v + 30 u v^2 + 4 u v^3 + 21 u^2 v + 15 u^2 v^2 + 3 u^2 v^3 + 3 u^3 v + 3 u^3 v^2 + u^3 v^3 + 9 u^2 + u^3$ & \multicolumn{2}{l}{18} \\
    4 & $3834 u + 5600 u v + 2168 u v^2 + 440 u v^3 + 38 u v^4 + 608 u^2 v + 506 u^2 v^2 + 160 u^2 v^3 + 19 u^2 v^4 + 60 u^3 v + 72 u^3 v^2 + 36 u^3 v^3 + 6 u^3 v^4 + 4 u^4 v + 6 u^4 v^2 + 4 u^4 v^3 + u^4 v^4 + 243 u^2 + 18 u^3 + u^4$ & \multicolumn{2}{l}{1606} \\
    5 & $1027080 u + 1377312 u v + 429320 u v^2 + 90280 u v^3 + 11920 u v^4 + 728 u v^5 + 51730 u^2 v + 43480 u^2 v^2 + 15160 u^2 v^3 + 2850 u^2 v^4 + 230 u^2 v^5 + 2375 u^3 v + 3130 u^3 v^2 + 1890 u^3 v^3 + 515 u^3 v^4 + 55 u^3 v^5 + 130 u^4 v + 220 u^4 v^2 + 180 u^4 v^3 + 70 u^4 v^4 + 10 u^4 v^5 + 5 u^5 v + 10 u^5 v^2 + 10 u^5 v^3 + 5 u^5 v^4 + u^5 v^5 + 20790 u^2 + 675 u^3 + 30 u^4 + u^5$ & \multicolumn{2}{l}{565080} \\
    6 & $1067308488 u + 1294038720 u v + 266378992 u v^2 + 45374320 u v^3 + 6287480 u v^4 + 588624 u v^5 + 26704 u v^6 + 14994792 u^2 v + 11427262 u^2 v^2 + 3403800 u^2 v^3 + 689230 u^2 v^4 + 87168 u^2 v^5 + 5098 u^2 v^6 + 237720 u^3 v + 315795 u^3 v^2 + 198240 u^3 v^3 + 62625 u^3 v^4 + 10920 u^3 v^5 + 825 u^3 v^6 + 6730 u^4 v + 12195 u^4 v^2 + 11180 u^4 v^3 + 5395 u^4 v^4 + 1290 u^4 v^5 + 125 u^4 v^6 + 240 u^5 v + 525 u^5 v^2 + 600 u^5 v^3 + 375 u^5 v^4 + 120 u^5 v^5 + 15 u^5 v^6 + 6 u^6 v + 15 u^6 v^2 + 20 u^6 v^3 + 15 u^6 v^4 + 6 u^6 v^5 + u^6 v^6 + 6364170 u^2 + 67635 u^3 + 1485 u^4 + 45 u^5 + u^6$ & \multicolumn{2}{l}{734774776} \\
    ... & ... & \multicolumn{2}{l}{...} \\
    % strongly connected tournament
    \hline\hline
    $n$ & $\bar g_n(u-1, v; (\sentence_{dg})_{TN}, \One, \One, R)$ & \multicolumn{2}{l}{$sct_n$} \\
    \hline
    1 & $u + u v$ & \multicolumn{2}{l}{1} \\
    2 & $2 u + 3 u v + u v^2$ & \multicolumn{2}{l}{0} \\
    3 & $8 u + 13 u v + 6 u v^2 + u v^3$ & \multicolumn{2}{l}{2} \\
    4 & $64 u + 98 u v + 43 u v^2 + 10 u v^3 + u v^4$ & \multicolumn{2}{l}{24} \\
    5 & $1024 u + 1398 u v + 465 u v^2 + 105 u v^3 + 15 u v^4 + u v^5$ & \multicolumn{2}{l}{544} \\
    6 & $32768 u + 39956 u v + 8488 u v^2 + 1495 u v^3 + 215 u v^4 + 21 u v^5 + u v^6$ & \multicolumn{2}{l}{22320} \\
    ... & ... & \multicolumn{2}{l}{...} \\
    % acyclicity and connected acyclicity
    \hline\hline
    $n$ & $g_n(u-1, v; \sentence_{dg}, \One, \One, R)$ & $ac_n$ & $cac_n$ \\
    \hline
    1 & $u + uv$ & 1 & 1 \\
    2 & $uv + uv^2 + 2u^2v + u^2v^2 + u^2$ & 3 & 2 \\
    3 & $-uv + 3uv^2 + 4uv^3 + 3u^2v + 6u^2v^2 + 3u^2v^3 + 3u^3v + 3u^3v^2 + u^3v^3 + u^3$ & 25 & 18 \\
    4 & $11uv - 19uv^2 + 8uv^3 + 38uv^4 - 4u^2v + 11u^2v^2 + 34u^2v^3 + 19u^2v^4 + 6u^3v + 18u^3v^2 + 18u^3v^3 + 6u^3v^4 + 4u^4v + 6u^4v^2 + 4u^4v^3 + u^4v^4 + u^4$ & 543 & 446 \\
    5 & $-363uv + 695uv^2 - 170uv^3 - 500uv^4 + 728uv^5 + 55u^2v - 50u^2v^2 - 35u^2v^3 + 300u^2v^4 + 230u^2v^5 - 10u^3v + 25u^3v^2 + 135u^3v^3 + 155u^3v^4 + 55u^3v^5 + 10u^4v + 40u^4v^2 + 60u^4v^3 + 40u^4v^4 + 10u^4v^5 + 5u^5v + 10u^5v^2 + 10u^5v^3 + 5u^5v^4 + u^5v^5 + u^5$ & 29281 & 26430 \\
    6 & $32157uv - 66863uv^2 + 24850uv^3 + 43190uv^4 - 53976uv^5 + 26704uv^6 - 2178u^2v + 2167u^2v^2 + 2970u^2v^3 - 4175u^2v^4 + 2298u^2v^5 + 5098u^2v^6 + 165u^3v - 15u^3v^2 - 210u^3v^3 + 1110u^3v^4 + 1965u^3v^5 + 825u^3v^6 - 20u^4v + 45u^4v^2 + 380u^4v^3 + 670u^4v^4 + 480u^4v^5 + 125u^4v^6 + 15u^5v + 75u^5v^2 + 150u^5v^3 + 150u^5v^4 + 75u^5v^5 + 15u^5v^6 + 6u^6v + 15u^6v^2 + 20u^6v^3 + 15u^6v^4 + 6u^6v^5 + u^6v^6 + u^6$ & 3781503 & 3596762 \\
    ... & ... & ... & ... \\
    % directed tree
    \hline\hline
    $n$ & $g_n(0, v; \sentence_{dt}, \One, \One, R)$ & \multicolumn{2}{l}{$dt_n$} \\
    \hline
    1 & $1 + v$ & \multicolumn{2}{l}{1} \\
    2 & $v + v^2$ & \multicolumn{2}{l}{2} \\
    3 & $(-1/2)v + (3/2)v^2 + 2v^3$ & \multicolumn{2}{l}{9} \\
    4 & $- 4v^2 + 2v^3 + 6v^4$ & \multicolumn{2}{l}{64} \\
    5 & $(7/2)v + (15/2)v^2 - 25v^3 - 5v^4 + 24v^5$ & \multicolumn{2}{l}{625} \\
    6 & $-25v + (51/2)v^2 + 130v^3 + (-285/2)v^4 - 102v^5 + 120v^6$ & \multicolumn{2}{l}{7776} \\
    ... & ... & \multicolumn{2}{l}{...}
  \end{longtable}
\end{center}

\section{Conclusion}

In this paper, we propose a novel perspective to investigate the weighted first-order model counting problem with axioms through the lens of graph polynomials.
We introduce three polynomials, \wcp{}, \nscp{} and \sscp{}, to address the challenge of incorporating axioms into the model counting problem.
With these polynomials and their efficient computation algorithms, we recover all existing complexity results on axioms and obtain new results on k-connected-components, bipartite, strongly connected, and spanning subgraph axioms, as well as certain combinations of these axioms.
Moreover, we demonstrate a tight connection between the weak connectedness polynomial (resp.\ strong connectedness polynomials) and the Tutte polynomial (resp. directed chromatic polynomials), enabling us to prove that any graph permits an efficient algorithm for computing its Tutte polynomial and directed chromatic polynomials as long as it can be expressed by a set of ground unary literals, cardinality constraints and a fixed \Ctwo{} sentence.
% \lnote{may need an update here.}
Our approach contributes to
%the interdisciplinary advancement of
both first-order model counting and enumerative combinatorics.
%, and we believe that it will benefit the development of both fields.

A natural direction for future work involves identifying additional domain-liftable axioms with our proposed polynomials.
%Expanding the applicability of the introduced polynomials to more generalized scenarios, such as the multi-relation version of \wcp{} and \scp{}s, where the relations are not necessarily nested, presents another avenue for exploration.
Additionally, within the field of enumerative combinatorics, using the proposed polynomials to discover intriguing integer sequences, similar to those found in the OEIS, and in~\cite{fluffy}, also holds a substantial interest.
Finally, from the perspective of complexity theory, our approach offers a unified framework for studying the upper bounds, i.e., the domain-liftability, of the model counting problem with axioms. However, the lower bounds for this problem remain open. It is worth exploring whether there exists a general criterion for determining the hardness of the model counting problem with axioms.

\subsection*{Acknowledgements}

Ond\v{r}ej Ku\v{z}elka's work was supported by Czech Science Foundation project 24-11820S (``Automatic Combinato\-rialist''). Yuanhong Wang's work was supported by National Natural Science Foundation of China (No.62506141). Yuyi Wang's work was partially supported by the Natural Science Foundation of Hunan Province, China (Grant No. 2024JJ5128).

% Contribution

% Future work

% \begin{itemize}
%     \item Apply this method to help the counting problems in combinatorics.
%     \item fluffy?
% \end{itemize}

% Hard problems:

% \begin{itemize}
%     \item It is still an open question to give a general criteria for the axioms that can be added to the predicates preserving the domain-liftability in WFOMC.
%     \item Our approach for multi-relational axioms only works for nested relations.
%     \item Acyclic axiom. (May not write this one)
% \end{itemize} 

\bibliographystyle{alpha}
\bibliography{ref}

\end{document}